\documentclass{article}
\usepackage[top=3cm, bottom=3cm, right=3cm, left=3cm]{geometry}
\usepackage[utf8]{inputenc}
\usepackage{graphicx}
\usepackage{mathrsfs}
\usepackage{physics}
\usepackage{braket}
\usepackage{hyperref}
\usepackage{float}
\usepackage{tikz-cd}
\usepackage{amsthm, amsmath, amssymb}
\usepackage{authblk}
\newcommand{\mathbbm}[1]{\text{\usefont{U}{bbm}{m}{n}#1}} 

\newcommand{\eqalign}[1]{\begin{equation}\begin{aligned}#1\end{aligned}\end{equation}}

\DeclareMathOperator{\Hom}{Hom}

\newtheorem{lemma}{Lemma}
\newtheorem{theorem}{Theorem}
\newtheorem{remark}{Remark}
\newtheorem{corollary}{Corollary}

\DeclareMathOperator{\Ker}{Ker}

\title{Entanglement Sum Rule from Higher-Form Symmetries}
\author[1]{Pei-Yao Liu\\ Email: \href{mailto:peiyaoliu@iphy.ac.cn}{peiyaoliu@iphy.ac.cn}}
\affil[1]{\textit{Institute of Physics, Chinese Academy of Sciences, Beijing 100190, China}}
\date{}
\begin{document}
\maketitle
\begin{abstract}
    We prove an entanglement sum rule for $(d-1)$-dimensional quantum lattice models with finite abelian higher-form symmetries, obtained by minimally coupling a sector on $p$-simplices carrying a $p$-form $G$ symmetry to a sector on $(p+1)$-simplices carrying the dual $(d-p-2)$-form $\widehat G$ symmetry (with $\widehat G$ being the Pontryagin dual of $G$). The coupling is introduced by conjugation with a symmetry-preserving operator $\mathcal{U}$ that dresses symmetry-invariant operators with appropriate Wilson operators. Our main result concerns symmetric eigenstates of the coupled model that arise by acting with $\mathcal{U}$ on direct-product symmetric eigenstates of the decoupled model: provided a topological criterion formulated via the Mayer--Vietoris sequence holds for the chosen bipartition, $\mathcal{U}$ factorizes across the cut when acting on the symmetric state, and the entanglement entropy equals the sum of the entropies of the two sectors. This framework explains and generalizes known examples in fermion-$\mathbb{Z}_2$ gauge theory, identifies when topology obstructs the sum rule, and provides a procedure to construct new examples by gauging higher-form symmetries.
\end{abstract}

\section{Introduction}
Entanglement sum rules capture when the entanglement entropy of a coupled quantum system decomposes additively into contributions from simpler subsystems. Besides offering a practical shortcut for computing entanglement, such additivity is a sharp diagnostic of how symmetry and topology restrict quantum correlations across a cut. They were discovered in models with free fermions coupled to $\mathbb{Z}_2$ gauge field: the entanglement entropy of the ground state equals the sum of fermionic and gauge field entanglements of the decoupled system \cite{Swingle2013,Yao_2010}. 

The original derivations relied on a 0-form symmetry of fermions (even parity) and a $(d-2)$-form symmetry of the $\mathbb{Z}_2$ gauge field (no gauge-flux), where $d$ is the dimension of spacetime. In general, a $p$-form symmetry acts on extended operators and can be generated by topological defects supported on codimension-$(p+1)$ submanifolds in spacetime \cite{Gaiotto_2015,McGreevy2023}. In lattice Hamiltonian realizations, the operators that implement these symmetries are labeled by $p$-cocycles or $(d-p-1)$-cycles of the dual complex, which are related by Poincar\'e duality. Minimal coupling can be implemented by dressing symmetric “matter” operators with Wilson operators. For the $p=0$ case, this is exactly the gauging map that promotes a global symmetry to a local one \cite{Haegeman_2015}. For the more general $p\ge 0$ case, the Wilson operator turns out to be the “bulk” of a $p$-boundary. These ingredients point to a purely homological mechanism for the entanglement sum rule.

\subsection{Main Results}
We establish a generalized entanglement sum rule for quantum lattice models with finite abelian higher-form symmetries. Consider a model with two decoupled sectors on a $(d-1)$-dimensional $\Delta$-complex $X$ ($d\ge 2$): a model on $p$-simplices with a $p$-form symmetry of a finite abelian group $G$, and a model on $(p+1)$-simplices with the dual $(d-p-2)$-form symmetry of the Pontryagin dual $\widehat{G}$ ($p=0,...,d-2$). We introduce minimal coupling between two sectors by conjugation with an operator $\mathcal{U}$ that dresses operators on $p$-simplices (or $(p+1)$-simplices) invariant under the $p$-form (or $(d-p-2)$-form) symmetries with Wilson operators built from $(p+1)$-simplices (or $p$-simplices). Conjugation with $\mathcal{U}$ preserves both higher-form symmetries. On the subspace invariant under both higher-form symmetries, $\mathcal{U}$ is well-defined and unitary, and the coupled Hamiltonian is related to the decoupled one by
\[
H \;=\; \mathcal{U}\,(H_p \otimes \mathbbm{1}_{p+1} + \mathbbm{1}_p \otimes H_{p+1})\,\mathcal{U}^\dagger.
\]
Our main theorem is that if a bipartition $(A,B)$ satisfies a Mayer--Vietoris criterion expressed in terms of homology groups for $A$, $B$, and their boundary, then $\mathcal{U}$ is equivalent to the tensor product of two local unitaries in the respective subregions when acting on a state with both symmetries. Consequently, for any direct-product symmetric eigenstate $\ket{\psi}_p\otimes\ket{\psi}_{p+1}$ of the decoupled model, the corresponding symmetric eigenstate of the minimally coupled model, $\mathcal{U}(\ket{\psi}_p\otimes\ket{\psi}_{p+1})$, obeys
\[
S_{\mathcal{H}_A}\left(\mathcal{U}(\ket{\psi}_p\otimes\ket{\psi}_{p+1})\right)
= S_{\mathcal{H}_{p,A}}(\ket{\psi}_p) + S_{\mathcal{H}_{p+1,A}}(\ket{\psi}_{p+1}).
\]
$S_*$ denotes the entanglement entropy. The details regarding how we decompose the Hilbert space with respect to a bipartition are given in Section \ref{subsec sum rule}. 

When the Mayer--Vietoris condition fails, topology obstructs the factorization of $\mathcal{U}$, and the sum rule can break down. If the higher-form symmetries are energetically enforced by stabilizer-like terms, the ground state subspace is spanned by symmetric eigenstates satisfying the entanglement sum rule. If the coupled model is viewed as a gauge theory, our result follows the extended-Hilbert-space notion of entanglement \cite{Aoki2015,Donnelly_2012}.

This framework recovers sum rules in fermion-$\mathbb{Z}_2$ gauge systems \cite{Swingle2013,Yao_2010} and demonstrates that the sum rule derives from symmetries rather than the detailed form of the Hamiltonian. This allows the construction of new examples by gauging higher-form symmetries in quantum lattice models \cite{Haegeman_2015, Choi_2025}.

\subsection{Outline}
Sections \ref{subsec AG}–\ref{subsec higher-form} collect the algebraic topology and symmetry preliminaries. Section \ref{subsec U} constructs the minimal coupling operator $\mathcal{U}$ and explains its relation to gauging. Section \ref{subsec sum rule} proves the factorization criterion and the sum rule. Section \ref{subsec review} revisits the fermion-$\mathbb{Z}_2$ gauge case through this lens. A new example is constructed in Section \ref{subsec new example} based on the $(3+1)$D transverse-field toric code.

\section{Proof of the Generalized Entanglement Sum Rule}\label{section 2}
\subsection{Simplicial (Co)homology: Preliminaries and Notation}\label{subsec AG}
This subsection introduces notation and recalls basic definitions used throughout the paper \cite{Hatcher}; it is not a comprehensive introduction. An $n$-simplex of a $\Delta$-complex $X$ is denoted by an ordered $(n+1)$-tuple $\sigma=[v_0,...,v_n]$ of its vertices. We denote the set of $n$-simplices of $X$ as $\Delta^n(X)$. The boundary map $\partial_n$ is defined as
\eqalign{
\partial_n [v_0,...,v_n]=\sum_{i=0}^n (-1)^i[v_0,...,\hat{v}_i,...,v_n].
}
The hat over $v_i$ means $v_i$ is removed from the tuple. $\partial_n$ maps an $n$-simplex to an alternating, formal sum of $(n-1)$-simplices. We set $\partial_0[v_0]=0$. It can be verified that $\partial_{n-1}\circ \partial_n=0$. 

Given an abelian group $G$, the chain group $C_n(X;G)$ consists of formal sums $\{\sum_{\sigma\in\Delta^n(X)} g_{\sigma} \sigma\mid g_{\sigma}\in G\}$. It is isomorphic to a direct sum $\bigoplus_{\sigma\in\Delta^{n}(X)} G$. In particular, $C_n(X;\mathbb{Z})$ is denoted simply by $C_n(X)$. For $n<0$ or $n>\dim X$, the chain group $C_n(X;G)$ is trivial (it only has the identity element). The boundary map $\partial_n$ induces a group homomorphism between two chain groups,
\eqalign{
\partial_n:C_n(X;G)\rightarrow C_{n-1}(X;G),\ \partial_n \left(\sum_{\sigma \in \Delta^n(X)} g_{\sigma}\sigma\right) = \sum_{\sigma \in \Delta^n(X)} g_{\sigma}\partial_n\sigma.
}
We define the $n$-cycles $Z_n(X;G)=\Ker\partial_n$ and $(n-1)$-boundaries $B_{n-1}(X;G)=\Im\partial_n$. So $\partial_{n}\circ \partial_{n+1}=0$ implies $B_n(X;G)\le Z_n(X;G)$. The $n$th homology group is defined as $H_n(X;G)=Z_n(X;G)/B_n(X;G)$.

Given an abelian group $G$, the homomorphisms from $C_n(X)$ to $G$ form an abelian group with component-wise addition; it is called the cochain group, denoted by $C^n(X;G)$. Given $f\in C^n(X;G)$, $f\circ \partial_{n+1} \in C^{n+1}(X;G)$. Thus, the boundary map $\partial_{n+1}: C_{n+1}(X)\rightarrow C_n(X)$ induces the coboundary map $\delta_n: C^n(X;G)\rightarrow C^{n+1}(X;G)$. It can be verified that $\delta_n$ is a homomorphism between cochain groups. 

We define the $n$-cocycles $Z^n(X;G)=\Ker\delta_n$ and $(n+1)$-coboundaries $B^{n+1}(X;G)=\Im\delta_{n}$. $\forall f\in C^{n-1}(X;G)$, $ \delta_{n}\circ\delta_{n-1} (f)=f \circ \partial_{n}\circ\partial_{n+1}=0$, thus $B^n(X;G)\le Z^n(X;G)$. The $n$th cohomology group is defined as $H^n(X;G)=Z^n(X;G)/B^n(X;G)$.

In the rest of the paper, we omit the subscripts of the boundary and coboundary maps when the source and target are clear from the context.

\subsection{Physical Setting and Symmetry Data}\label{subsec physical}
These assumptions apply in the rest of the paper unless stated otherwise. Let $d\ge 2$ be the spacetime dimension. The space $X$ is a finite $(d-1)$-dimensional $\Delta$-complex corresponding to an orientable closed manifold (a regular CW complex should also work, as it also admits a dual complex). The Hilbert space $\mathcal{H}$ is the tensor product of local Hilbert spaces on $ \Delta^p(X)$ and $ \Delta^{p+1}(X)$ ($p=0,...,d-2$),
\eqalign{
    &\mathcal{H}_{p}=\bigotimes_{\sigma\in \Delta^p(X)}\mathcal{H}_{\sigma},\\ &\mathcal{H}_{p+1}=\bigotimes_{\sigma\in \Delta^{p+1}(X)}\mathcal{H}_{\sigma},\\ &\mathcal{H}=\mathcal{H}_p\otimes\mathcal{H}_{p+1}.
}

$G$ is a finite abelian group. $\{U_{\sigma}(g)\mid g\in G\}$ is a unitary representation of $G$ on $\mathcal{H}_{\sigma},\ \sigma\in\Delta^p(X)$.

The Pontryagin dual $\widehat{G}= \Hom(G,\mathbb{R}/2\pi\mathbb{Z})$ is the group of irreducible representations of $G$. Addition in $\widehat{G}$ is defined component-wise,
\eqalign{
    &(\rho_1+\rho_2)(g)=\rho_1(g)+\rho_2(g).
}
Since $G$ is a finite abelian group, $G\cong \widehat{G}$ \cite{HewittRoss1979}. 

$\{\tilde{U}_{\sigma}(\rho )\mid \rho \in \widehat{G}\}$ is a unitary representation of $\widehat{G}$ on $\mathcal{H}_{\sigma},\ \sigma\in \Delta^{p+1}(X)$.
 
\subsection{Higher-Form Symmetries: Operators and Representations}\label{subsec higher-form}
Given $\phi \in C^p(X;G)$, define
\eqalign{
    U(\phi)=\prod_{\sigma \in\Delta^p(X)}U_{\sigma}(\phi(\sigma)).
}
The $p$-form $G$ symmetry transformations on $\mathcal{H}_p$ are defined as 
\eqalign{\label{p-form}
    \{U(\phi)\mid \phi\in Z^p(X;G)\}.
}
Physically, these operators represent conserved quantities associated with closed $(d-1-p)$-dimensional submanifolds in the dual complex \cite{Gaiotto_2015, Luo_2024, Sakura2023}, as seen via Poincar\'e duality: Let $\bar{X}$ be the dual complex of $X$. There exists a natural isomorphism $D:C^{n}(X;G)\rightarrow C_{d-1-n}(\bar{X};G),\ \phi\mapsto \bar{\phi}$ induced by intersections of simplices $\sigma\leftrightarrow\bar{\sigma}$. $D$ satisfies $D\circ\delta=\partial \circ D$. Thus, $D(Z^{p}(X;G))=Z_{d-1-p}(\bar{X};G)$, the symmetry transformations can be equivalently written as
\eqalign{\label{p-form2}
    \{\prod_{\bar{\sigma}\in \Delta^{d-1-p}(\bar{X})}U_{\sigma}(\bar{\phi}_{\bar{\sigma}})\mid \bar{\phi}\in Z_{d-1-p}(\bar{X};G)\}.
}
(Note that the symmetry transformations still act on $\mathcal{H}_p$; only the index set has been dualized. The same applies to later expressions involving dual complexes.) Indeed, the $(d-1-p)$-cycles may be viewed as closed $(d-1-p)$-dimensional submanifolds labeled by $G$ elements.

The projector onto an irreducible representation of the local unitary operators on $ \Delta^p(X)$ corresponding to an element of $C_p(X;\widehat{G})$ is defined as
\eqalign{
    &P(k)=\prod_{\sigma\in \Delta^p(X)}\frac{1}{\abs{G}}\sum_{g\in G} e^{-ik_{\sigma}(g)} U_{\sigma}(g).}
It is straightforward to verify that $P$ is a projector. Note that we don't require $\{U_{\sigma}(g)\mid g\in G\}$ to contain each irreducible representation of $G$, so $P(k)$ might equal 0 for $k\neq 0$. The projectors satisfy $\forall \phi \in C^p(X;G),$
\eqalign{
    U(\phi) P(k)&=\prod_{\sigma\in\Delta^p(X)} \frac{1}{\abs{G}}\sum_{g\in G}e^{-ik_{\sigma}(g)} U_{\sigma}(\phi(\sigma)g)\\
    &=\prod_{\sigma\in\Delta^p(X)} e^{ik_{\sigma}(\phi(\sigma))}\frac{1}{\abs{G}}\sum_{g\in G}e^{-ik_{\sigma}(g)} U_{\sigma}(g)\\
    &=e^{i\braket{k,\phi}}P(k),
}
where
\eqalign{
    \braket{\cdot,\cdot}: C_n(X;\widehat{G})\times C^n(X;G)\rightarrow \mathbb{R}/2\pi\mathbb{Z},\  (k,\phi)\mapsto \sum_{\sigma\in \Delta^n(X)}k_{\sigma}(\phi(\sigma)),
}
for $n=0,1,...,d-1$. $\braket{\cdot, \cdot}$ is a perfect pairing between $C_n(X;\widehat{G})$ and $C^n(X;G)$ and satisfies a generalized Stokes' theorem \cite{Munkres1984},
\eqalign{\label{Stokes}
\braket{\partial k,\phi}=\braket{k,\delta\phi}.
}
See Appendix \ref{appendixA} for more details.

It can be verified that $\sum_{k\in C_p(X;\widehat{G})}P(k)=1$ and $P(k_1)P(k_2)=\delta_{k_1,k_2}$ using the Schur orthogonality relations \cite{Serre1977},
\eqalign{
    &\frac{1}{\abs{G}}\sum_{\rho\in \widehat{G}} e^{-i\rho(g_1g_2^{-1})}=\delta_{g_1,g_2},\\
    &\frac{1}{\abs{G}}\sum_{g\in G} e^{i(\rho_1-\rho_2)(g)}=\delta_{\rho_1,\rho_2}.
}

\begin{lemma}\label{lemma1}
The subspace of $\mathcal{H}_p$ that is invariant under the $p$-form symmetry transformations \eqref{p-form} is the range of 
\eqalign{ \label{p-form invariant}
\sum_{k\in B_p(X;\widehat{G})}P(k).
}
\end{lemma}
\begin{proof}
    Since $\sum_{k\in C_p(X;\widehat{G})}P(k)=1$, $U(\phi)=\sum_{k\in C_p(X;\widehat{G})} e^{i\braket{k,\phi}} P(k)$. Thus, $\forall k\in C_p(X;\widehat{G})$, the range of $P(k)$ is invariant under all $p$-form symmetry transformations if and only if $\forall \phi \in Z^p(X;G)$, $\braket{k,\phi} = 0\mod 2\pi$. It is proved in Appendix \ref{appendixA} that this is equivalent to $k\in B_p(X;\widehat{G})$.
\end{proof}
A simple example regarding the 1-form $\mathbb{Z}_2$ symmetry, where $d=3$, is illustrated in Figure~\ref{fig:z2example}.
\begin{figure}[H]
    \centering
    \includegraphics[width=0.4\linewidth]{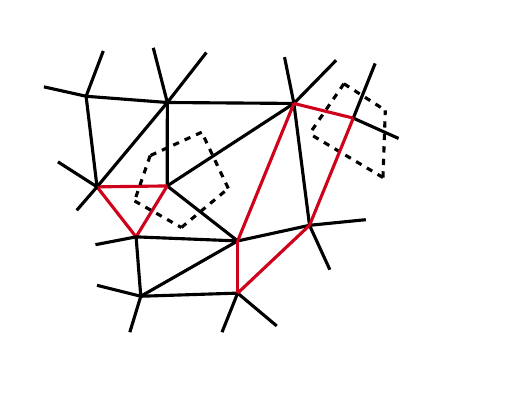}
    \caption{An example of the 1-form $\mathbb{Z}_2$ symmetry where $d=3$. The space $X$ is a two-dimensional complex. The closed black dashed lines on the dual complex represent $\phi\in Z^1(X;\mathbb{Z}_2)\cong Z_1(\bar{X};\mathbb{Z}_2)$: For a link $\sigma\in\Delta^1(X)$, $\phi(\sigma)=1$ if $\sigma$ crosses a dashed line, otherwise $\phi(\sigma)=0$. $\phi$ represents a 1-form symmetry transformation $U(\phi)=\prod_{\sigma\in\Delta^1(X)}U_{\sigma}(\phi(\sigma))$, where $U_{\sigma}(0)=U^2_{\sigma}(1)=1$. The red lines represent $k\in C_1(X;\mathbb{Z}_2)$: $k_{\sigma}=1$ if the link $\sigma$ is colored red, otherwise $k_{\sigma}=0$. $k$ represents an irreducible representation of local unitary operators, the corresponding projector is $P(k)=\prod_{\sigma\in\Delta^1(X)} P_{\sigma}(k_{\sigma})$. $P_{\sigma}(1)$ is the projection to the sign representation where $U_{\sigma}(0)=1, U_{\sigma}(1)=-1$; $P_{\sigma}(0)$ is the projection to the trivial representation where $U_{\sigma}(0)=U_{\sigma}(1)=1$. Then we know by acting $U(\phi)$ on the subspace that $P(k)$ projects onto, we get $e^{i\braket{k,\phi}}$ where $\braket{k,\phi}$ equals the number of intersections of the black dashed lines and red lines times $\pi$. Thus, $P(k)$ projects onto a subspace that is invariant under all 1-form symmetry transformations if and only if the red lines have an even number of intersections with all closed loops on the dual complex. Lemma \ref{lemma1} yields that this is equivalent to $k$ being a boundary, or equivalently, the red lines being the domain walls between two sets of plaquettes. This is indeed the case in the figure.}
    \label{fig:z2example}
\end{figure}
\begin{corollary}
    The irreducible representations of the $p$-form symmetry transformations are labeled by elements of $C_p(X;\widehat{G})/B_p(X;\widehat{G})$. The projector onto the subspace corresponding to $[k]\in C_p(X;\widehat{G})/B_p(X;\widehat{G})$ is
    \eqalign{
    \sum_{k'-k\in B_p(X;\widehat{G})} P(k').
    }
    It satisfies
    \eqalign{
    &U(\phi)\sum_{k'-k\in B_p(X;\widehat{G})} P(k')=e^{i\braket{k,\phi}}\sum_{k'-k\in B_p(X;\widehat{G})} P(k').
    }
    Since $(B^p(X;G))^{\perp}=Z_p(X;\widehat{G})$ (see Appendix \ref{appendixA}), an irreducible representation of the $p$-form symmetry transformations that is invariant under the “local” ones, $
    \{U(\phi)\mid \phi\in B^p(X;G)\}$, is labeled by a homology class $[k]\in Z_p(X;\widehat{G})/B_p(X;\widehat{G})=H_p(X;\widehat{G})$.
\end{corollary}

Given $k=\sum_{\sigma\in\Delta^{p+1}(X)}k_{\sigma}\sigma\in C_{p+1}(X;\widehat{G})$, define
\eqalign{
    \tilde{U}(k)=\prod_{\sigma\in\Delta^{p+1}(X)}\tilde{U}_{\sigma}(k_{\sigma}).
}
The “dual” $(d-p-2)$-form $\widehat{G}$ symmetry transformations in $\mathcal{H}_{p+1}$ are defined as (compare with \eqref{p-form2})
\eqalign{\label{(d-p-2)-form}
    \{\tilde{U}(k)\mid k \in Z_{p+1}(X;\widehat{G})\}.
}
Similarly, we have an isomorphism $D':C_n(X;\widehat{G})\rightarrow C^{d-1-n}(\bar{X};\widehat{G}),\ k\mapsto \bar{k}$ via Poincar\'e duality. And $D'(Z_{p+1}(X;\widehat{G}))=Z^{d-p-2}(\bar{X};\widehat{G})$. Thus, these symmetry transformations can be equivalently written as (compare with \eqref{p-form})
\eqalign{
    \{\prod_{\bar{\sigma}\in \Delta^{d-p-2}(\bar{X})}\tilde{U}_{\sigma}(\bar{k}(\bar{\sigma}))\mid \bar{k}\in Z^{d-p-2}(\bar{X};\widehat{G})\}.
}
Similar to $P(k)$, we can define the projector onto an irreducible representation of the local unitaries on $ \Delta^{p+1}(X)$ corresponding to $\phi'\in C^{p+1}(X;G)$,
\eqalign{
    &\tilde{P}(\phi')=\abs{G}^{-\abs{ \Delta^{p+1}(X)}}\sum_{k'\in C_{p+1}(X;\widehat{G})}e^{-i\braket{k',\phi'}} \tilde{U}(k'),\\
    &\forall k'\in C_{p+1}(X;\widehat{G}),\ \tilde{U}(k') \tilde{P}(\phi')=e^{i\braket{k',\phi'}}\tilde{P}(\phi').
}
Since $(Z_{p+1}(X;\widehat{G}))^{\perp}=B^{p+1}(X;G)$ (see Appendix \ref{appendixA}), the projector onto the subspace invariant under all $(d-p-2)$-form $\widehat{G}$ symmetry transformations is
\eqalign{\label{d-p-2 invariant}
\sum_{\phi'\in B^{p+1}(X;G)} \tilde{P}(\phi').
}
We denote the subspace of $\mathcal{H}$ that is invariant under the $p$-form \eqref{p-form} and $(d-p-2)$-form \eqref{(d-p-2)-form} symmetry transformations as $\mathcal{H}_{\text{inv}}$. The projector onto $\mathcal{H}_{\text{inv}}$ is
\eqalign{
P_{\text{inv}}=\sum_{k\in B_p(X;\widehat{G})} P(k)\sum_{\phi'\in B^{p+1}(X;G)} \tilde{P}(\phi').
}

\subsection{The Minimal Coupling Operator and Gauging}\label{subsec U}
We define the minimal coupling operator that couples $\mathcal{H}_{p}$ and $\mathcal{H}_{p+1}$,
\eqalign{\label{unitary}
    \mathcal{U} = \sum_{k=\partial k'\in B_p(X;\widehat{G})} P(k) \tilde{U}(k').
}
Note that $\mathcal{U}$ commutes with the $p$-form and $(d-p-2)$-form symmetry transformations. The choice of the preimage $k'$ of $k$ is not unique. If $\partial k' = \partial k'' = k$, then $\tilde{U}(k'')$ and $\tilde{U}(k')$ differ by $\tilde{U}(k'-k'')$. Since $k''-k'\in Z_{p+1}(X;\widehat{G})$, $\tilde{U}(k'-k'')$ is a $(d-p-2)$-form symmetry transformation \eqref{(d-p-2)-form}. Thus, $\mathcal{U}$ is well-defined (independent of choices of $k'$) when restricted to the subspace invariant under the $(d-p-2)$-form symmetry transformations. Since $\mathcal{U}\mathcal{U}^{\dagger}=\mathcal{U}^{\dagger}\mathcal{U}=\sum_{k\in B_p(X;\widehat{G})} P(k)$, by Lemma \ref{lemma1}, $\mathcal{U}$ is unitary when restricted to the subspace invariant under the $p$-form symmetry transformations. In summary, when restricted to $\mathcal{H}_{\text{inv}}$, $\mathcal{U}$ is a well-defined unitary operator.

To see the connection between $\mathcal{U}$ and the minimal coupling prescription in lattice gauge theory, we may view $\mathcal{H}_p$ as the matter field and $\mathcal{H}_{p+1}$ as the gauge field. To define gauge transformations corresponding to the gauged $p$-form symmetry, we need one more assumption in this subsection: Each irreducible representation of $\{\tilde{U}_{\sigma}(\rho )\mid \rho \in\widehat{G}\}$ is contained and has the same multiplicity in $\mathcal{H}_{\sigma}$, $\forall\sigma\in\Delta^{p+1}(X)$. This assumption is only for the interpretation from the lattice gauge theory perspective, but is unnecessary for the entanglement sum rule.

We now define the gauge transformations corresponding to the gauged $p$-form symmetry. For a $(p+1)$-simplex $\sigma\in \Delta^{p+1}(X)$, since the irreducible representations of $\widehat{G}$ are assumed to have the same multiplicity in $\mathcal{H}_{\sigma}$, we can define (non-canonically) $\tilde{T}_{\sigma}(g)$ that sends the subspace corresponding to $g'\in G$ to that corresponding to $gg'\in G$, so that
\eqalign{
    \tilde{T}_{\sigma}(g) \tilde{U}_{\sigma}(\rho)=e^{-i\rho(g)}\tilde{U}_{\sigma}(\rho)\tilde{T}_{\sigma}(g).
}
For example, when $\widehat{G}\cong G = \mathbb{Z}_2$, we may take $T_{\sigma}(0)=\text{id}_{\sigma}$, $T_{\sigma}(1)=\sigma^z_{\sigma}$, $U_{\sigma}(0)=\text{id}_{\sigma}$, $U_{\sigma}(1)=\sigma^x_{\sigma}$. 

For $\phi'\in C^{p+1}(X;G)$, define $\tilde{T}(\phi')=\prod_{\sigma\in \Delta^{p+1}(X)} \tilde{T}_{\sigma} (\phi'(\sigma))$, then
\eqalign{\label{heisenberg}
    \tilde{T}(\phi') \tilde{U}(k')=e^{-i\braket{k',\phi'}}\tilde{U}(k')\tilde{T}(\phi').
}
Actually, $\{\tilde{T}(\phi')\mid \phi'\in Z^{p+1}(X;G)\}$ are the closed 't Hooft operators; $\{\tilde{U}(k')\mid k'\in Z_{p+1}(X;\widehat{G})\}$ are the closed Wilson operators.

The gauge transformations are
\eqalign{\label{gauge transformation}
    \{U(\phi)\tilde{T}(\delta \phi)\mid \phi\in C^p(X;G)\}.
}
Note that $\phi$ is no longer required to be closed.

Next, we see how conjugation with $\mathcal{U}$ transforms an operator on $\mathcal{H}_p$ commuting with the $p$-form symmetry transformations into an operator on $\mathcal{H}_p\otimes\mathcal{H}_{p+1}$ commuting with the gauge transformations. Each operator on $\mathcal{H}_p$ can be expanded as
\eqalign{
    O = \sum_{k\in C_p(X;\widehat{G})}P(k) O \sum_{k\in C_p(X;\widehat{G})}P(k)=\sum_{q\in C_p(X;\widehat{G})} \left[\sum_{k\in C_p(X;\widehat{G})}P(k) O P(k-q)\right]=\sum_{q\in C_p(X;\widehat{G})} O_q.
}
$O_q$ shifts the $p$-chain label of the representation by $q$: $\forall k \in C_p(X;\widehat{G}),\ P(k)\,O_q=O_q\,P(k-q)$. $\forall \phi\in Z^p(X;G),\ q\in C_p(X;\widehat{G})$, we have
\eqalign{\label{component}
    U(\phi) O_q &=U(\phi)\sum_{k\in C_p(X;\widehat{G})} P(k) O_q= \sum_{k\in C_p(X;\widehat{G})}e^{i\braket{k,\phi}} P(k) O_q\\
    &= O_q\sum_{k\in C_p(X;\widehat{G})}e^{i\braket{k,\phi}} P(k-q)=e^{i\braket{q,\phi}} O_q U(\phi).
}
From $(Z^p(X;G))^{\perp} = B_p(X;\widehat{G})$ (see Appendix \ref{appendixA}), $O_q$ commutes with each $p$-form symmetry transformation if and only if $q\in B_p(X;\widehat{G})$. By straightforward calculations,

\eqalign{\label{gauged}
    \mathcal{U}O\mathcal{U}^{\dagger} = \sum_{q\in C_p(X;\widehat{G}) }\mathcal{U} O_q \mathcal{U}^{\dagger} =\sum_{q=\partial q'\in B_p(X;\widehat{G})} O_q \tilde{U}(q') \sum_{k\in B_p(X;\widehat{G})} P(k).
}
We now verify that $\mathcal{U} O \mathcal{U}^{\dagger}$ commutes with all gauge transformations. Using Stokes' theorem \eqref{Stokes}, the commutation relation \eqref{heisenberg}, and the fact that $\tilde{T}(\delta\phi)$ commutes with operators acting trivially on $\mathcal{H}_{p+1}$,
\eqalign{\label{commute}
    &U(\phi)\tilde{T}(\delta\phi) \mathcal{U} O \mathcal{U}^{\dagger}\\
    =&\sum_{k=\partial k' \in B_p(X;\widehat{G})} \left[e^{i\braket{k,\phi}-i\braket{k',\delta\phi}}P(k) \tilde{U}(k')\right] \tilde{T}(\delta \phi)O \mathcal{U}^{\dagger}\\
    =&\sum_{k=\partial k' \in B_p(X;\widehat{G})} \left[e^{i\braket{k,\phi}-i\braket{\partial k',\phi}}P(k) \tilde{U}(k')\right]\tilde{T}(\delta \phi)O \mathcal{U}^{\dagger}\\
    =& \mathcal{U} O\tilde{T}(\delta\phi)\mathcal{U}^{\dagger}= \mathcal{U}O \sum_{k=\partial k' \in B_p(X;\widehat{G})}  \tilde{T}(\delta\phi)P(k) \tilde{U}(-k')\\
    =& \mathcal{U}O \sum_{k=\partial k' \in B_p(X;\widehat{G})} P(k) \tilde{U}(-k') e^{-i\braket{k,\phi}+i\braket{k',\delta\phi} } U(\phi)\tilde{T}(\delta\phi)\\
    =&\mathcal{U}  O \mathcal{U}^{\dagger}U(\phi)\tilde{T}(\delta \phi).
}
Indeed, $\mathcal{U}$ transforms $O$ into a gauge-invariant operator. In $\mathcal{U}O\mathcal{U}^{\dagger}$ \eqref{gauged}, every symmetric component, $O_q$, of the bare matter operator $O$ is preserved and dressed by the Wilson operator $\tilde{U}(q')$ corresponding to a $(p+1)$-chain, whose boundary is the change of the irreducible representation of local unitaries induced by $O_q$, which is exactly the lattice form of minimal coupling. The choice of preimage $q'$ is not unique. Suppose $\partial q' = \partial q''=q$, then their corresponding Wilson operators differ by $\tilde{U}(q'-q'')$. Since $q'-q''\in Z_{p+1}(X;\widehat{G})$, $\tilde{U}(q'-q'')$ is a $(d-p-2)$-form symmetry transformation \eqref{(d-p-2)-form}. Thus, in the subspace invariant under the $(d-p-2)$-form symmetry transformations, the choices of $q'$ are equivalent. This non-uniqueness of minimal coupling is familiar in gauge theories \cite{Haegeman_2015, Jenkins_2013}. For example, in fermion-$\mathbb{Z}_2$ gauge theory, to make $c^{\dagger}_{v_1} c_{v_2}$ invariant under the generators of gauge transformations, $(-1)^{c^{\dagger}_v c_v}\prod_{l\ni v}\sigma^x_l$, we can dress it with a product of $\sigma^z_l$ operators along any path from $v_2$ to $v_1$, as illustrated in Figure~\ref{fig:ambiguity}.
\begin{figure}
    \centering
    \includegraphics[width=0.6\linewidth]{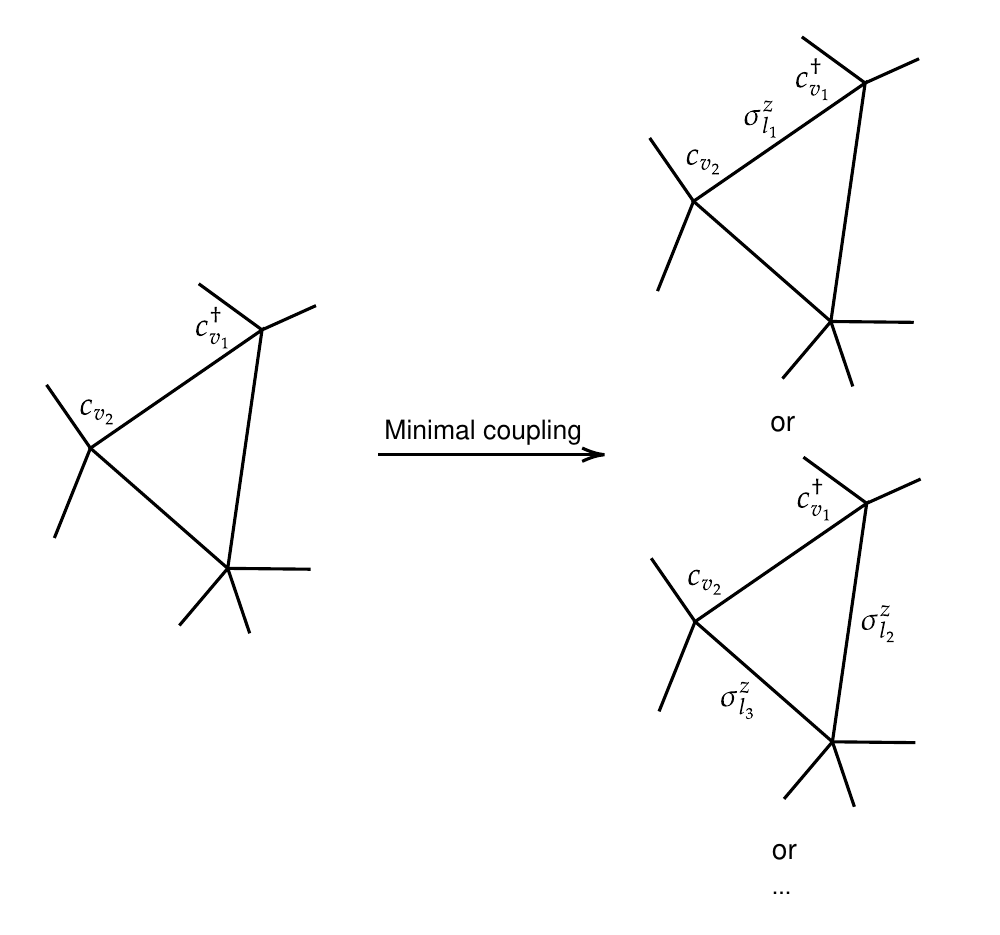}
    \caption{An example of the ambiguity of minimal coupling in fermion-$\mathbb{Z}_2$ gauge theory.}
    \label{fig:ambiguity}
\end{figure}
But under the restriction of no $\mathbb{Z}_2$ flux (both local and global), all choices of paths become equivalent.

It was introduced in \cite{Haegeman_2015} how to promote a global (0-form) symmetry of an operator to a local (gauge) symmetry by gauging. For a $p$-form $G$ symmetry on $\mathcal{H}_p$, we can achieve this through conjugation with $\mathcal{U}$.

To begin with, we enlarge the Hilbert space $\mathcal{H}_p$ to $\mathcal{H}_p\otimes \mathcal{H}_{p+1}$, where $\mathcal{H}_{p+1}$ is the tensor product of canonical representations of $G$ on each $(p+1)$-simplex,
\eqalign{
\tilde{T}_{\sigma}(g')\ket{g}_{\sigma} = \ket{g'g}_{\sigma},\ \tilde{U}_{\sigma}(\rho)\ket{g}_{\sigma}=e^{i\rho(g)}\ket{g}_{\sigma},
 \ \sigma\in\Delta^{p+1}(X).
}
To gauge a state $\ket{\psi}_p\in \mathcal{H}_p$, we embed it in the enlarged Hilbert space $\mathcal{H}_p\otimes\mathcal{H}_{p+1}$, and then sum over all gauge transformations \eqref{gauge transformation},
\eqalign{
\hat{G}\ket{\psi}_p=\frac{1}{\abs{G}^{\abs{ \Delta^{p}(X)}}}\sum_{\phi\in C^{p}(X;G)} U(\phi)\tilde{T}(\delta\phi)\left(\ket{\psi}_p \bigotimes_{\sigma\in\Delta^{p+1}(X)} \ket{\text{id}}_{\sigma}\right),
}
where $\text{id}$ denotes the identity element of $G$. (The gauging map $\hat{G}$ has a small hat over $G$ to distinguish it from the Pontryagin dual $\widehat{G}$.) Note that 
\eqalign{
&\frac{1}{\abs{G}^{\abs{ \Delta^{p}(X)}}}\sum_{\phi\in C^{p}(X;G)} U(\phi)\tilde{T}(\delta\phi)\\
=&\frac{1}{\abs{G}^{\abs{ \Delta^{p}(X)}}}\sum_{\phi\in C^{p}(X;G)} U(\phi)\tilde{T}(\delta\phi)\frac{1}{\abs{Z^p(X;G)}}\sum_{\phi\in Z^p(X;G)} U(\phi)\\
=&\frac{1}{\abs{G}^{\abs{ \Delta^{p}(X)}}}\sum_{\phi\in C^{p}(X;G)} U(\phi)\tilde{T}(\delta\phi)\sum_{k\in B_p(X;\widehat{G})} P(k).
}
Thus, the component of $\ket{\psi}_p$ orthogonal to the subspace invariant under the $p$-form symmetry is projected out by $\hat{G}$.

The gauged version of an operator $O$ on $\mathcal{H}_p$ is denoted by $\mathcal{G}[O]$. $\mathcal{G}[O]$ should commute with all gauge transformations \eqref{gauge transformation} and be compatible with $\hat{G}$,
\eqalign{
\hat{G} O\ket{\psi}_p=\mathcal{G}[O]\hat{G}\ket{\psi}_p
}
for any $\ket{\psi}_p$ with the $p$-form symmetry \cite{Haegeman_2015}. We will show that $\mathcal{G}[O]=\mathcal{U} O \mathcal{U}^{\dagger}$ satisfies these conditions.

We have demonstrated that $O=\sum_{q\in C_p(X;\widehat{G})}O_q$ where $O_q=\sum_{k\in C_p(X;\widehat{G})}P(k)OP(k-q)
$. Since $\hat{G}$ contains a projection onto the subspace invariant under the $p$-form symmetry transformations, $O_q\ket{\psi}_p$ for $q\notin B_p(X;\widehat{G})$ is projected out. Thus,
\eqalign{
\hat{G}O\ket{\psi}_p=&\frac{1}{\abs{G}^{\abs{ \Delta^{p}(X)}}}\sum_{\phi\in C^{p}(X;G)} U(\phi)\tilde{T}(\delta\phi)\left(O\ket{\psi}_p\bigotimes_{\sigma\in\Delta^{p+1}(X)}\ket{\text{id}}_{\sigma}\right)\\
=&\frac{1}{\abs{G}^{\abs{ \Delta^{p}(X)}}}\sum_{\phi\in C^{p}(X;G)} U(\phi)\tilde{T}(\delta\phi)\sum_{q\in B_p(X;\widehat{G})}O_q\left(\ket{\psi}_p\bigotimes_{\sigma\in\Delta^{p+1}(X)}\ket{\text{id}}_{\sigma}\right).
}
Since $\tilde{U}(q')$ acts trivially on $\bigotimes_{\sigma\in\Delta^{p+1}(X)}\ket{\text{id}}_{\sigma}$ and $\sum_{k\in B_p(X;\widehat{G})} P(k)$ acts trivially on $\ket{\psi}_p$,
\eqalign{
&\hat{G}O\ket{\psi}_p\\
=&\frac{1}{\abs{G}^{\abs{ \Delta^{p}(X)}}}\sum_{\phi\in C^{p}(X;G)} U(\phi)\tilde{T}(\delta\phi)\sum_{q=\partial q'\in B_p(X;\widehat{G})} O_q \tilde{U}(q')\sum_{k\in B_p(X;\widehat{G})} P(k)\left(\ket{\psi}_p\bigotimes_{\sigma\in\Delta^{p+1}(X)}\ket{\text{id}}_{\sigma}\right)\\
=&\frac{1}{\abs{G}^{\abs{ \Delta^{p}(X)}}}\sum_{\phi\in C^{p}(X;G)} U(\phi)\tilde{T}(\delta\phi)\mathcal{U} O \mathcal{U}^{\dagger}\left(\ket{\psi}_p\bigotimes_{\sigma\in\Delta^{p+1}(X)}\ket{\text{id}}_{\sigma}\right).
}
Since $\mathcal{U} O \mathcal{U}^{\dagger}$ commutes with all gauge transformations \eqref{commute},
\eqalign{
\hat{G} O\ket{\psi}_p&=\mathcal{U} O \mathcal{U}^{\dagger}\frac{1}{\abs{G}^{\abs{ \Delta^{p}(X)}}}\sum_{\phi\in C^{p}(X;G)} U(\phi)\tilde{T}(\delta\phi)\left(\ket{\psi}_p\bigotimes_{\sigma\in\Delta^{p+1}(X)}\ket{\text{id}}_{\sigma}\right)\\
&=\mathcal{U} O \mathcal{U}^{\dagger} \hat{G}\ket{\psi}_p.
}
Indeed, $\mathcal{U}$ gauges the operator $O$ with the $p$-form symmetry.

Alternatively, we may straightforwardly extend the original construction \cite{Haegeman_2015} to the $p$-form $G$ symmetry,
\eqalign{
\mathcal{G}'[O]=\frac{1}{\abs{Z^p(X;G)}}\sum_{\phi\in C^p(X;G)} 
   U(\phi)\,\tilde{T}(\delta\phi)
   \left(
      O \bigotimes_{\sigma\in \Delta^{p+1}(X)} 
      \ket{\text{id}}\bra{\text{id}}_{\sigma}
   \right)
   U^{\dagger}(\phi)\,\tilde{T}^{\dagger}(\delta\phi).
}
 By setting $p=0$, we return to the original construction for a finite abelian global symmetry. Note that the original construction also works for non-abelian groups \cite{Haegeman_2015}. It can be verified that $\mathcal{G}'[O]\hat{G}\ket{\psi}_p = \hat{G} O\ket{\psi}_p$ if $\ket{\psi}_p$ is invariant under $p$-form symmetry transformations. However, there are some subtle differences between the two prescriptions. Consider an arbitrary operator $O=\sum_{q\in C_p(X;\widehat{G})}O_q$ on $\mathcal{H}_p$,
\eqalign{
\mathcal{G}'[O]
=&\frac{1}{\abs{Z^p(X;G)}}\sum_{\phi\in C^p(X;G)} U(\phi)\tilde{T}(\delta\phi) \left(O\bigotimes_{\sigma\in\Delta^{p+1}(X)} \ket{\text{id}}\bra{\text{id}}_{\sigma}\right)U^{\dagger}(\phi)\tilde{T}^{\dagger}(\delta\phi)\\
=&\frac{1}{\abs{Z^p(X;G)}}\sum_{q\in C_p(X;\widehat{G})}\sum_{\phi\in C^p(X;G)} e^{i\braket{q,\phi}}O_q\otimes\ket{\delta\phi}\bra{\delta\phi}_{p+1}.
}
For $q\notin B_p(X;\widehat{G})$, since $(Z^p(X;G))^{\perp}=B_p(X;\widehat{G})$,
\eqalign{
&\sum_{\phi\in C^p(X;G)} e^{i\braket{q,\phi}}O_q\otimes\ket{\delta\phi}\bra{\delta\phi}_{p+1}\\
=&\sum_{\phi\in C^p(X;G)}e^{i\braket{q,\phi}}\frac{1}{\abs{Z^p(X;G)}}\sum_{\phi_1\in Z^p(X;G)} e^{i\braket{q,\phi_1}} O_q\otimes\ket{\delta\phi}\bra{\delta\phi}_{p+1}=0.
}
Thus,
\eqalign{
\mathcal{G}'[O]=&\frac{1}{\abs{Z^p(X;G)}}\sum_{q=\partial q'\in B_p(X;\widehat{G})} \sum_{\phi\in C^p(X;G)} e^{i\braket{q',\delta\phi}} O_q\otimes\ket{\delta\phi}\bra{\delta\phi}_{p+1}\\
=&\sum_{q=\partial q'\in B_p(X;\widehat{G})} O_q \tilde{U} (q')\sum_{\phi'\in B^{p+1}(X;G)}\ket{\phi'}\bra{\phi'}_{p+1}\\
=&\sum_{q=\partial q'\in B_p(X;\widehat{G})} O_q \tilde{U} (q')\sum_{\phi'\in B^{p+1}(X;G)}\tilde{P}(\phi').
}
The last term $\sum_{\phi'\in B^{p+1}(X;G)}\tilde{P}(\phi')$ is the projector onto the subspace invariant under the $(d-p-2)$-form symmetry transformations. Comparing with \eqref{gauged}, we observe that the two prescriptions are equivalent when restricted to the symmetric subspace $\mathcal{H}_{\text{inv}}$. 

\iffalse
Though this prescription also satisfies the requirement for gauging, $\mathcal{G}[O]\hat{G}\ket{\psi}_p=\hat{G}O\ket{\psi}_p$, the global projection on $\mathcal{H}_1$ is unnatural. For example, using the original construction, gauging the identity operator $1_{\mathcal{H}_0}$ gives $1_{\mathcal{H}_0}\bigotimes_{\sigma\in \Delta^1(X)} \ket{\text{id}}\bra{\text{id}}_{\sigma}$. While using the minimal coupling operator $\mathcal{U}$, gauging the identity operator gives $\sum_{k\in B_0(X;\widehat{G})}P(k)$, which equals identity when restricted to the gauge-invariant subspace. Note that $\mathcal{U}$ is still defined without restricting to the subspace with the $p$-form and $(d-p-2)$-form symmetry. The differences are the minimal coupling prescription depends on the choices of the Wilson operators, and $\mathcal{U}$ won't be unitary. So no generality is lost in our construction.
\fi

Another nice property of the minimal coupling operator $\mathcal{U}$ is that the matter-gauge duality appears explicitly when restricted to $\mathcal{H}_{\text{inv}}$.
We switch the roles of $\mathcal{H}_p$ and $\mathcal{H}_{p+1}$ via Poincar\'e duality and define 
\eqalign{
    \bar{\mathcal{U}}=\sum_{\bar{\phi}'=\partial \bar{\phi}\in B_{d-p-2}(\bar{X};G)}\tilde{P}(\phi')U(\phi)=\sum_{\phi'=\delta\phi\in B^{p+1}(X;G)}\tilde{P}(\phi')U(\phi).
}
It is verified in Appendix \ref{appendixB} that $\mathcal{U}=\bar{\mathcal{U}}$ in $\mathcal{H}_{\text{inv}}$. Consequently, given an operator on $\mathcal{H}_{p+1}$ with the $(d-p-2)$-form $\widehat{G}$ symmetry, conjugation by $\mathcal{U}$ dresses it with the corresponding Wilson operator on $\mathcal{H}_p$.

\subsection{Factorization of $\mathcal{U}$ and Entanglement Sum Rule}\label{subsec sum rule}
\begin{figure}[H]
    \centering
    \includegraphics[width=0.9\linewidth]{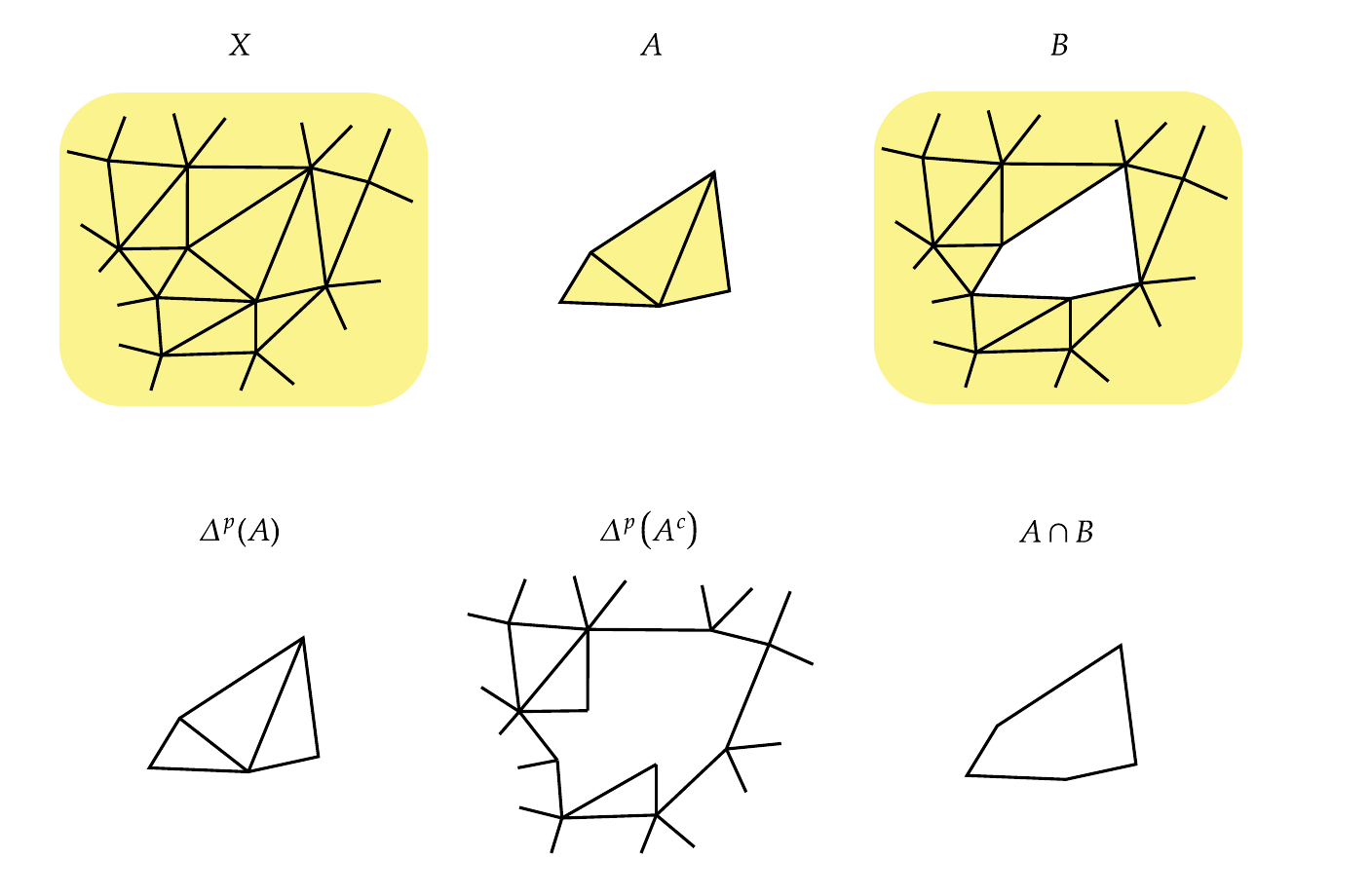}
    \caption{An illustration of the bipartition in the case where $d=3$ and $p=1$. The 1-simplices are the links. The 2-simplices are the triangles colored yellow.}
    \label{fig:factorization}
\end{figure}
Next, let's specify how we decompose the Hilbert space when calculating the entanglement entropy. Let $A$ be a $(p+1)$-dimensional subcomplex of $X$. Let $B$ be the $(p+1)$-dimensional subcomplex of $X$ whose $(p+1)$-simplices are $ \Delta^{p+1}(X)\setminus \Delta^{p+1}(A)$. $A\cap B$ is the boundary consisting of $p$-simplices. $\Delta^p(A^c)= \Delta^p(X)\setminus \Delta^p(A)$. Denote the $p$-chains generated by $ \Delta^p(A^c)$ as $C_p (A^c;\widehat{G})$, so that $C_p(X;\widehat{G})=C_p(A;\widehat{G})\oplus C_p(A^c;\widehat{G})$. 

Let $\mathcal{H}_A$ be the tensor product of the local Hilbert spaces on all $p$- and $(p+1)$-simplices in $A$, and $\mathcal{H}_{A^c}$ be the tensor product of the local Hilbert spaces on other $p$- and $(p+1)$-simplices. Then $\mathcal{H}=\mathcal{H}_{A}\otimes\mathcal{H}_{A^c}$, where
\eqalign{
    &\mathcal{H}_{A}=\mathcal{H}_{p,A}\otimes \mathcal{H}_{p+1,A}= (\bigotimes_{\sigma\in\Delta^p(A)}\mathcal{H}_{\sigma})\otimes (\bigotimes_{\sigma\in\Delta^{p+1}(A)}\mathcal{H}_{\sigma}),\\
    &\mathcal{H}_{A^c}=\mathcal{H}_{p,A^c}\otimes \mathcal{H}_{p+1,B}= (\bigotimes_{\sigma\in\Delta^p(A^c)}\mathcal{H}_{\sigma})\otimes (\bigotimes_{\sigma\in\Delta^{p+1}(B)}\mathcal{H}_{\sigma}).
}
A special case, where $d=3$ and $p=1$, is illustrated in Figure~\ref{fig:factorization}. 

Given a state $\ket{\psi}\in\mathcal{H}$, the reduced density matrix on $\mathcal{H}_A$ is $\rho_{\mathcal{H}_A}=\Tr_{\mathcal{H}_{A^c}}(\ket{\psi}\bra{\psi})$, and the entanglement entropy is
\eqalign{
S_{\mathcal{H}_A} (\ket{\psi})= -\Tr_{\mathcal{H}_A}(\rho_{\mathcal{H}_A} \ln \rho_{\mathcal{H}_A}).
}

Given a state $\ket{\psi}_p\in \mathcal{H}_p$, the reduced density matrix on $\mathcal{H}_{p,A}$ is $\rho_{\mathcal{H}_{p,A}} = \Tr_{\mathcal{H}_{p,A^c}}(\ket{\psi}_p\bra{\psi}_p)$, and the entanglement entropy is
\eqalign{
S_{\mathcal{H}_{p,A}}(\ket{\psi}_p) = -\Tr_{\mathcal{H}_{p,A}}(\rho_{\mathcal{H}_{p,A}} \ln \rho_{\mathcal{H}_{p,A}}).
}

Given a state $\ket{\psi}_{p+1}\in \mathcal{H}_{p+1}$, the reduced density matrix on $\mathcal{H}_{p+1,A}$ is $\rho_{\mathcal{H}_{p+1,A}} = \Tr_{\mathcal{H}_{p+1,B}}(\ket{\psi}_{p+1}\bra{\psi}_{p+1})$, and the entanglement entropy is
\eqalign{
S_{\mathcal{H}_{p+1,A}}(\ket{\psi}_{p+1}) = -\Tr_{\mathcal{H}_{p+1,A}}(\rho_{\mathcal{H}_{p+1,A}} \ln \rho_{\mathcal{H}_{p+1,A}}).
}

Given the $(p+1)$-subcomplexes $A$ and $B$, we have an exact sequence,
\eqalign{
&H_{p}(A\cap B;\widehat{G})\overset{\begin{bmatrix}
    i_A^*\\
    -i_B^*
\end{bmatrix}
}{\longrightarrow} H_p(A;\widehat{G})\oplus H_p(B;\widehat{G})\overset{\begin{bmatrix}j_A^*,j_B^*\end{bmatrix}}{\longrightarrow}H_p(A\cup B;\widehat{G}).
}
This is a segment of the Mayer--Vietoris sequence \cite{Hatcher}. $i_A^*$ and $i_B^*$ are induced by inclusion maps from $A\cap B$ to $A$ and $B$. $j_A^*$ and $j_B^*$ are induced by inclusion maps from $A$ and $B$ to $A\cup B$ \cite{Hatcher}.

We say $\Im(\begin{bmatrix}
    i_A^*\\
    -i_B^*
\end{bmatrix})$ is the graph of an isomorphism if there is an isomorphism
\eqalign{\kappa:\Im(i_A^*)\rightarrow \Im(-i_B^*),}
with $\Im(\begin{bmatrix}
    i_A^*\\
    -i_B^*
\end{bmatrix})=\{(x,\kappa(x))\mid x\in \Im(i_A^*)\}$.

\begin{theorem}\label{theorem1}
If $\Im(\begin{bmatrix}
    i_A^*\\
    -i_B^*
\end{bmatrix})$ is the graph of an isomorphism, $\mathcal{U} P_{\text{inv}}=(\mathcal{U}_{A}\otimes \mathcal{U}_{A^c} )P_{\text{inv}}$, with $\mathcal{U}_A$ and $\mathcal{U}_{A^c}$ acting on $\mathcal{H}_A$ and $\mathcal{H}_{A^c}$ respectively. (The expressions of $\mathcal{U}_A$ and $\mathcal{U}_{A^c}$ will be given in the proof.)
\end{theorem}
Intuitively, this theorem states that when restricted to $\mathcal{H}_{\text{inv}}$, $\mathcal{U}$ effectively factorizes to each subregion. Thus, we can independently add minimal coupling in $A$ and $B$ subregions, preventing entanglement generation across the cut and leading to the sum rule.

\begin{proof}
    Since $C_p(X;\widehat{G})=C_p(A;\widehat{G})\oplus C_p(A^c;\widehat{G})$, $k\in C_p(X;\widehat{G})$ has a unique decomposition $k_A + k_{A^c}$. For each $k\in B_p(X;\widehat{G})$, we aim to independently find $f_1(k_A)\in C_{p+1}(A;\widehat{G})$ and $f_2(k_{A^c})\in C_{p+1}(B;\widehat{G})$ such that $\partial(f_1(k_A)+f_2(k_{A^c}))=k$. The functions $f_1$ and $f_2$ will be used to construct $\mathcal{U}_A$ and $\mathcal{U}_{A^c}$. 
    
    Suppose $k_A+k_{A^c} = \partial k'$ for some $k'\in C_{p+1}(X;\widehat{G})$. Since $C_{p+1}(X;\widehat{G})=C_{p+1}(A;\widehat{G})\oplus C_{p+1}(B;\widehat{G})$, $k'$ can also be uniquely decomposed as $k'_A+k'_B$. Then $k_A- \partial k'_A= -(k_{A^c}-\partial k'_B)\in C_p(A\cap B;\widehat{G})$. Thus, $\partial k_A= -\partial k_{A^c}=\partial(k_A-\partial k'_A)\in B_{p-1}(A\cap B;\widehat{G})$.

    We can define 
    \eqalign{h:B_{p-1}(A\cap B;\widehat{G})\rightarrow C_{p}(A\cap B;\widehat{G})}
    by choosing a preimage for each $(p-1)$-boundary, so that $\partial\circ h=\text{id}$. Although $h$ may not be a homomorphism, we can make our choices to ensure $h(-k'')=-h(k'')$, so that $h(\partial k_A)+h(\partial k_{A^c})=0$. Then $\partial(k_A - h(\partial k_A))=0$ and $\partial( k_{A^c}-h(\partial k_{A^c}))=0$, and the corresponding homology classes are $[k_A-h(\partial k_A)]\in H_p(A;\widehat{G})$, $[ k_{A^c}-h(\partial k_{A^c})]\in H_p(B;\widehat{G})$.

    Since $[k_A-h(\partial k_A)+k_{A^c}-h(\partial k_{A^c})]=[k]=0\in H_p(A\cup B;\widehat{G})$, 
    $$([k_A-h(\partial k_A)],[k_{A^c}-h(\partial k_{A^c})])\in\Ker[j_A^*, j_B^*]=\Im \begin{bmatrix}
        i_A^*\\
        -i_B^*
    \end{bmatrix}$$
    Thus, we may choose a $p$-cycle $t([k_A-h(\partial k_A)],[ k_{A^c}-h(\partial k_{A^c})])\in Z_p(A\cap B;\widehat{G})$ representing a preimage of $([k_A-h(\partial k_A)],[ k_{A^c}-h(\partial k_{A^c})])$ under $\begin{bmatrix}
        i_A^*\\
        -i_B^*
    \end{bmatrix}$. Since $\Im(\begin{bmatrix}
        i_A^*\\
        -i_B^*
    \end{bmatrix})$ is the graph of an isomorphism, $t([k_A-h(\partial k_A)],[ k_{A^c}-h(\partial k_{A^c})])$ is determined by either $[k_A-h(\partial k_A)]$ or $[ k_{A^c}-h(\partial k_{A^c})]$, with the other fixed by the isomorphism. Thus, we can define
    \eqalign{
    &t_A: \Im(i_A^*)\rightarrow Z_p(A\cap B;\widehat{G}),\\
    &t_B: \Im(-i_B^*)\rightarrow Z_p(A\cap B;\widehat{G}).
    }
    And $\forall k\in B_p(X;\widehat{G})$, $t([k_A-h(\partial k_A)],[ k_{A^c}-h(\partial k_{A^c})])=t_A([k_A-h(\partial k_A)])=t_B([k_{A^c}-h(\partial k_{A^c})])$. Then $k_A-h(\partial k_A)-t_A([k_A-h(\partial k_A)])\in B_p(A;\widehat{G})$ and $k_{A^c}-h(\partial k_{A^c})+t_B([k_{A^c}-h(\partial k_{A^c})])\in B_p(B;\widehat{G})$.
    
    Similarly, we can define the maps
    \eqalign{
        &l_A: B_p(A;\widehat{G}) \rightarrow C_{p+1}(A;\widehat{G}),\\
        &l_B: B_p(B;\widehat{G}) \rightarrow C_{p+1}(B;\widehat{G}),
        }
    by choosing a preimage for each $p$-boundary. 
    Then we can define the desired $f_1$ and $f_2$:
    \eqalign{
        &f_1: C'_p(A;\widehat{G})\rightarrow C_{p+1}(A;\widehat{G}),\ f_1(k_A)=l_A(k_A-h(\partial k_A)-t_A([k_A-h(\partial k_A)]),\\
        &f_2: C'_p(A^c;\widehat{G})\rightarrow C_{p+1}(B;\widehat{G}),\ f_2(k_{A^c})=l_B(k_{A^c}-h(\partial k_{A^c})+t_B([k_{A^c}-h(\partial k_{A^c})]),
    }
    where
    \eqalign{
    &C_p'(A;\widehat{G}) = \{k_A\in C_p(A;\widehat{G})\mid \partial k_A\in B_p(A\cap B;\widehat{G}), [k_A-h(\partial k_A)]\in \Im(i_A^*)\},\\
    &C_p'(A^c;\widehat{G}) = \{k_{A^c}\in C_p(A^c;\widehat{G})\mid \partial k_{A^c}\in B_p(A\cap B;\widehat{G}),[k_{A^c}-h(\partial k_{A^c})]\in \Im(-i_B^*)\}.
    }
    We have demonstrated that $\forall k=k_A+k_{A^c}\in B_p(X;\widehat{G})$, $k_A\in C_p'(A;\widehat{G})$, and $k_{A^c}\in C_p'(A^c;\widehat{G})$. Then
    \eqalign{
        \partial(f_1(k_A)+f_2(k_{A^c}))&=\partial(l_A(k_A-h(\partial k_A))+ l_B(k_{A^c}-h(\partial k_{A^c})))\\
        &=k_A+k_{A^c}-h(\partial k_A)-h(\partial k_{A^c})=k.
    }
    An example of how $f_1$ and $f_2$ work for the $p=1$ case is shown in Figure~\ref{fig:sum rule}.
    \begin{figure}[H]
        \centering
        \includegraphics[width=0.9\linewidth]{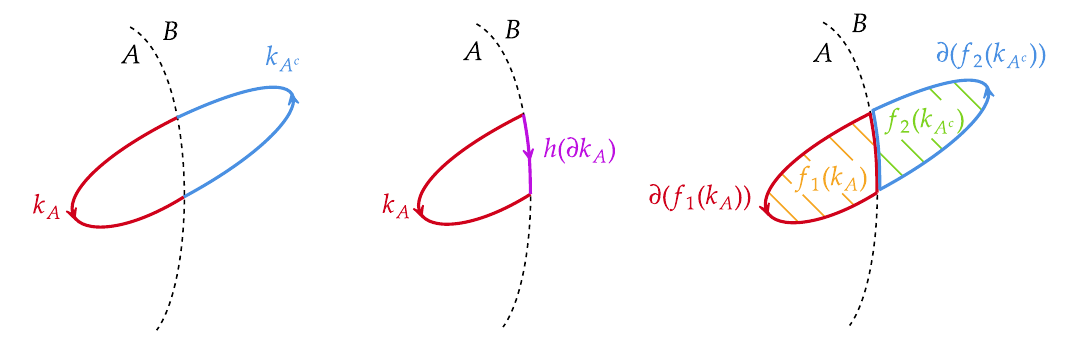}
        \caption{An illustration of $f_1$ and $f_2$ in a special case where $p=1$. $k=k_A+k_{A^c}$ is depicted on the left. In this case $t_A([k_A-h(\partial k_A)])=t_B([k_{A^c}-h(\partial k_{A^c})])=0$, since $k_A-h(\partial k_A)$ and $k_{A^c}-h(\partial k_{A^c})$ are already boundaries in $A$ and $B$ respectively.}
        \label{fig:sum rule}
    \end{figure}
    Using $f_1$ and $f_2$, we can define
    \eqalign{
        &\mathcal{U}_A = \sum_{k_A\in C_p'(A;\widehat{G})}P_A(k_A)\prod_{\sigma\in\Delta^{p+1}(A)}\tilde{U}_{\sigma}(f_1(k_A)_{\sigma}),\\
        &\mathcal{U}_{A^c}= \sum_{k_{A^c}\in C'_p(A^c;\widehat{G})}P_{A^c}(k_{A^c})\prod_{\sigma\in\Delta^{p+1}(B)}\tilde{U}_{\sigma}(f_2(k_{A^c})_{\sigma}),
    }
    where
    \eqalign{
        &P_A(k_A)=\prod_{\sigma\in \Delta^p(A)}\frac{1}{\abs{G}}\sum_{g\in G} e^{-ik_{A\sigma}(g)} U_{\sigma}(g),\\
        &P_{A^c}(k_{A^c})=\prod_{\sigma\in \Delta^p(A^c)}\frac{1}{\abs{G}}\sum_{g\in G} e^{-ik_{A^c\sigma}(g)} U_{\sigma}(g).
    }
    By Lemma \ref{lemma1}, $P_A(k_A)P_{A^c}(k_{A^c}) P_{\text{inv}}=0$ if $k_A+k_{A^c}\notin B_p(X;\widehat{G})$. Then, it is straightforward to check $\mathcal{U}P_{\text{inv}}=(\mathcal{U}_{A}\otimes \mathcal{U}_{A^c})P_{\text{inv}}$ by expanding both sides of the equality. Note that different choices of $h$, $t_A$, $t_B$, $l_A$, and $l_B$ are equivalent to different choices of the preimage of $k$ in $\mathcal{U}$, which are equivalent in $\mathcal{H}_{\text{inv}}$.
\end{proof}
\begin{remark}
    In particular, if $H_p(A;\widehat{G})$ and $H_p(B;\widehat{G})$ are trivial, the condition that $\Im(\begin{bmatrix}
    i_A^*\\
    -i_B^*
\end{bmatrix})$ is the graph of an isomorphism is vacuously true.
\end{remark}

Let $H_p$ be a Hamiltonian on $\mathcal{H}_p$ with the $p$-form $G$ symmetry, and $H_{p+1}$ be a Hamiltonian on $\mathcal{H}_{p+1}$ with the $(d-p-2)$-form $\widehat{G}$ symmetry. $H_0=H_p \otimes \mathbbm{1}_{p+1} + \mathbbm{1}_p \otimes H_{p+1}$. When restricted to the symmetric subspace $\mathcal{H}_{\text{inv}}$, $\mathcal{U}$ is a well-defined (independent of choices of $k'$) unitary operator and $H=\mathcal{U} H_0 \mathcal{U}^{\dagger}$ is the minimally coupled Hamiltonian. Thus, $\left\{\mathcal{U}\left(\ket{\psi}_p\otimes \ket{\psi}_{p+1}\right)\right\}$ are symmetric eigenstates of $H$ that span $\mathcal{H}_{\text{inv}}$, where $\left\{\ket{\psi}_p\right\}$ and $\left\{\ket{\psi}_{p+1}\right\}$ are symmetric eigenstates of $H_p$ and $H_{p+1}$ respectively. The following theorem yields the entanglement sum rule of such eigenstates of $H$.
\begin{theorem}\label{theorem2}
    For $\ket{\psi}_p\otimes\ket{\psi}_{p+1}\in\mathcal{H}_{\text{inv}}$, if $\Im(\begin{bmatrix}
    i_A^*\\
    -i_B^*
\end{bmatrix})$ is the graph of an isomorphism,
$$S_{\mathcal{H}_A}\left(\mathcal{U}\left(\ket{\psi}_p\otimes \ket{\psi}_{p+1}\right)\right)=S_{\mathcal{H}_{p,A}}\left(\ket{\psi}_p\right)+S_{\mathcal{H}_{p+1,A}}\left(\ket{\psi}_{p+1}\right).$$
\end{theorem}
\begin{proof}
    By Theorem \ref{theorem1}, if $\Im(\begin{bmatrix}
    i_A^*\\
    -i_B^*
\end{bmatrix})$ is the graph of an isomorphism, $\mathcal{U}\left(\ket{\psi}_p\otimes \ket{\psi}_{p+1}\right)=(\mathcal{U}_A\otimes \mathcal{U}_{A^c})(\ket{\psi}_p\otimes \ket{\psi}_{p+1})$. From the cyclic property of the trace, the reduced density matrix of $\mathcal{U}\left(\ket{\psi}_p\otimes \ket{\psi}_{p+1}\right)$ on $\mathcal{H}_A$ satisfies
    \eqalign{
    &\Tr_{\mathcal{H}_{A^c}}\left[\mathcal{U}\left(\ket{\psi}_p\bra{\psi}_p\otimes \ket{\psi}_{p+1}\bra{\psi}_{p+1}\right)\mathcal{U}^{\dagger}\right]\\
    =&\Tr_{\mathcal{H}_{A^c}} \left[ \mathcal{U}_A \otimes (\mathcal{U}_{A^c}^{\dagger} \mathcal{U}_{A^c}) \left(\ket{\psi}_p\bra{\psi}_p\otimes \ket{\psi}_{p+1}\bra{\psi}_{p+1}\right)\mathcal{U}_A^{\dagger} \right]\\
    =&\mathcal{U}_A\Tr_{\mathcal{H}_{A^c}} \left[\sum_{k_{A^c}\in C_p'(A^c;\widehat{G})} P_{A^c} (k_{A^c})\ket{\psi}_p\bra{\psi}_p\otimes \ket{\psi}_{p+1}\bra{\psi}_{p+1}  \right]\mathcal{U}_A^{\dagger}.
    }
    The definitions of $C_p'(A^c;\widehat{G})$ and $P_{A^c}(k_{A^c})$ were given in the proof of Theorem \ref{theorem1}. Since $\ket{\psi}_p$ is invariant under the $p$-form symmetry transformations, by Lemma \ref{lemma1},
    \eqalign{
    \sum_{k_{A^c}\in C_p'(A^c;\widehat{G})} P_{A^c} (k_{A^c})\ket{\psi}_p=\sum_{k_{A^c}\in C_p'(A^c;\widehat{G})} P_{A^c} (k_{A^c})\sum_{k\in B_p(X;\widehat{G})} P(k)\ket{\psi}_p.
    }
    It was demonstrated in the proof of Theorem \ref{theorem1} that each $k\in B_p(X;\widehat{G})$ can be uniquely decomposed as $k_A+k_{A^c}$ where $k_{A^c}\in C_p'(A^c;\widehat{G})$. Thus,
    \eqalign{
    \sum_{k_{A^c}\in C_p'(A^c;\widehat{G})} P_{A^c} (k_{A^c})\ket{\psi}_p = \sum_{k\in B_p(X;\widehat{G})} P(k)\ket{\psi}_p=\ket{\psi}_p,
    }
    and
    \eqalign{
    \Tr_{\mathcal{H}_{A^c}}\left[\mathcal{U}\left(\ket{\psi}_p\bra{\psi}_p\otimes \ket{\psi}_{p+1}\bra{\psi}_{p+1}\right)\mathcal{U}^{\dagger}\right]
    =\mathcal{U}_A \Tr_{\mathcal{H}_{A^c}}\left(\ket{\psi}_p\bra{\psi}_p\otimes \ket{\psi}_{p+1}\bra{\psi}_{p+1}\right)\mathcal{U}_A^{\dagger}.
    }
    Similarly,
    \eqalign{
    &S_{\mathcal{H}_A}\left(\mathcal{U}\left(\ket{\psi}_p\otimes \ket{\psi}_{p+1}\right)\right)\\
    =&-\Tr_{\mathcal{H}_A}\left\{\mathcal{U}_A \Tr_{\mathcal{H}_{A^c}}\left(\ket{\psi}_p\bra{\psi}_p\otimes \ket{\psi}_{p+1}\bra{\psi}_{p+1}\right)\mathcal{U}_A^{\dagger}\ln \left[\mathcal{U}_A \Tr_{\mathcal{H}_{A^c}}\left(\ket{\psi}_p\bra{\psi}_p\otimes \ket{\psi}_{p+1}\bra{\psi}_{p+1}\right)\mathcal{U}_A^{\dagger}\right]\right\}\\
    =&-\Tr_{\mathcal{H}_A}\left[\Tr_{\mathcal{H}_{A^c}}\left(\ket{\psi}_p\bra{\psi}_p\otimes \ket{\psi}_{p+1}\bra{\psi}_{p+1}\right)\ln \Tr_{\mathcal{H}_{A^c}}\left(\ket{\psi}_p\bra{\psi}_p\otimes \ket{\psi}_{p+1}\bra{\psi}_{p+1}\right)\right]\\
    =&S_{\mathcal{H}_A}\left(\ket{\psi}_p\otimes \ket{\psi}_{p+1}\right).
    }
    The direct-product form of $\ket{\psi}_p\otimes \ket{\psi}_{p+1}$ indicates
    \eqalign{
    S_{\mathcal{H}_A}\left(\mathcal{U}\left(\ket{\psi}_p\otimes \ket{\psi}_{p+1}\right)\right)=S_{\mathcal{H}_A}\left(\ket{\psi}_p\otimes \ket{\psi}_{p+1}\right)=S_{\mathcal{H}_{p,A}}\left(\ket{\psi}_p\right)+S_{\mathcal{H}_{p+1,A}}\left(\ket{\psi}_{p+1}\right).
    }
\end{proof}

\section{Examples of the Entanglement Sum Rule}\label{section example}
Conjugation by $\mathcal{U}$ implements the minimal coupling prescription between two models with respective symmetries. This explains the previously discovered entanglement sum rules in models with minimal coupling between fermions and 1-form $\mathbb{Z}_2$ gauge field. We can construct new examples of entanglement sum rules by gauging higher-form symmetries corresponding to any finite abelian group.
\subsection{Fermions Coupled to $\mathbb{Z}_2$ Gauge Field}\label{subsec review}
Let's see how our theorem recovers results in fermion models coupled to $\mathbb{Z}_2$ gauge field \cite{Swingle2013}. The decoupled Hamiltonian is
\eqalign{\label{decoupled1}
    H_0 = &-w\sum_{l\in\Delta^1(X)} (c^{\dagger}_{l^+}c_{l^-}+\text{h.c.})- \mu\sum_{v\in\Delta^0(X)} c^{\dagger}_{v}c_v \\
    &-J\sum_{l\in\Delta^1(X)}\sigma_l^z-gJ\sum_{v\in\Delta^0(X)} \prod_{l\ni v} \sigma_l^x - V\sum_{f\in\Delta^2(X)}\prod_{l\in f}\sigma^z_l,
}
where $\partial l= l^+-l^-$ if we view $l$ as an element of $C(X)$. This Hamiltonian has a 0-form symmetry corresponding to the fermion parity on each connected component and a $(d-2)$-form symmetry corresponding to the product of $\sigma^z_l$ on closed loops. The minimal coupling operator \eqref{unitary} should be written as \cite{Swingle2013}
\eqalign{
&\mathcal{U}= \sum_{k=\partial k'\in B_0(X;\mathbb{Z}_2)} P(k) \tilde{U}(k'),\\
&P(k) = \prod_{v\in \Delta^0(X)} \frac{1+(-1)^{k_{v}}(-1)^{c^{\dagger}_v c_v}}{2},\\
&\tilde{U}(k') = \prod_{l\in\Delta^1(X)} (\sigma^z_l)^{k'_{l}}.
}

We now derive $\mathcal{U} H_0 \mathcal{U}^{\dagger}$ restricted to the subspace $\mathcal{H}_{\text{inv}}$ invariant under the 0-form and $(d-2)$-form symmetry transformations. The $(d-2)$-form symmetry ensures $\mathcal{U}$ is independent of the choices of $k'$. All terms of $H_0$, except the $-w$ and $-gJ$ terms, commute with $\mathcal{U}$. For the $-w$ terms, observe that
\eqalign{
    P(k) c^{\dagger}_{l^+}c_{l^-} = c^{\dagger}_{l^+}c_{l^-} P(k-\partial l).
}
Thus, according to \eqref{gauged}, $q=l^+-l^-$, $q' = l$, we have
\eqalign{
\mathcal{U} c^{\dagger}_{l^+} c_{l^-} \mathcal{U}^{\dagger} = c^{\dagger}_{l^+}\sigma^z_l c_{l^-}.
}
Note that we omitted the projection operator $\sum_{k\in B_0(X;\mathbb{Z}_2)}P(k)$ since we are restricted to $\mathcal{H}_{\text{inv}}$.

The $-gJ$ terms commute with $P(k)$. If $k_{v}=0$, $P(k)$ projects onto the subspace with $c^{\dagger}_vc_v=0$. Correspondingly, $\tilde{U}(k')= \prod_{l\in\Delta^1(X)} (\sigma^z_l)^{k'_{l}}$ acts with $\sigma^z_l$ on an even number of neighboring links of $v$, then $\tilde{U}(k')\prod_{l\ni v}\sigma^x_l \tilde{U}^{\dagger}(k') = \prod_{l\ni v}\sigma^x_l$. If $k_{v}=1$, $P(k)$ projects onto the subspace with $c^{\dagger}_vc_v=1$ and $\tilde{U}(k')\prod_{l\ni v}\sigma^x_l \tilde{U}^{\dagger}(k') = -\prod_{l\ni v}\sigma^x_l$. Thus, $\mathcal{U}\prod_{l\ni v} \sigma^x_l \mathcal{U}^{\dagger} = (-1)^{c^{\dagger}_v c_v} \prod_{l\ni v}\sigma^x_l$. Again, $\sum_{k\in B_0(X;\mathbb{Z}_2)}P(k)$ is omitted since we are restricted to $\mathcal{H}_{\text{inv}}$.

In summary, in the symmetric subspace $\mathcal{H}_{\text{inv}}$,
\eqalign{
    \mathcal{U} H_0 \mathcal{U}^{\dagger} =& -w\sum_{l\in\Delta^1(X)} (c^{\dagger}_{l^+}\sigma_l^zc_{l^-}+\text{h.c.})-\mu\sum_{v\in\Delta^0(X)}c^{\dagger}_v c_v\\
    & - J\sum_{l\in\Delta^{1}(X)}\sigma_l^z-gJ\sum_{v\in\Delta^0(X)}(-1)^{c^{\dagger}_v c_v}\prod_{l\ni v}\sigma^x_l-V\sum_{f\in\Delta^{2}(X)}\prod_{l\in f}\sigma_l^z\\
    \equiv&H.
}
$H$ is defined in $\mathcal{H}$ with the same expression as above. 

As $p=0$, we consider $1$-dimensional subcomplexes $A$ and $B$ of $X$ and factorize the Hilbert space following Section \ref{subsec sum rule}. Consider the symmetric eigenstates of $H$, $\{\mathcal{U}(\ket{\psi}_p\otimes\ket{\psi}_{p+1})\}$, where $\{\ket{\psi}_p\}$ and $\{\ket{\psi}_{p+1}\}$ are symmetric eigenstates of the decoupled fermions and $\mathbb{Z}_2$ gauge field (the first and second lines of \eqref{decoupled1}). From Theorem \ref{theorem2}, if $\Im(\begin{bmatrix}
    i^*_A\\
    -i^*_B
\end{bmatrix})$ is the graph of an isomorphism, these eigenstates of $H$ satisfy the entanglement sum rule. In particular, a sufficient condition is that $A$ and $B$ are connected.

Next, let's discuss the condition of the ground state of $H$ being in the symmetric subspace, $\mathcal{H}_{\text{inv}}$. Observe that the plaquette operator $\prod_{l\in f}\sigma_l^z$ is a stabilizer that commutes with all other terms, when $V>0$, the ground state must be in the subspace where $\prod_{l\in f}\sigma_l^z=1,\ \forall f\in\Delta^{2}(X)$. If $H_{1}(X;\mathbb{Z}_2)$ is trivial, the plaquette operators generate all $(d-2)$-form symmetry transformations, and the ground state of $H$ must be invariant under the $(d-2)$-form symmetry transformations. For the fermion parity symmetry, although it is not energetically enforced, it could be tuned with the chemical potential.

To see more clearly how $\mathcal{U}$ factorizes in $\mathcal{H}_{\text{inv}}$, we can explicitly construct $\mathcal{U}_{A}$ and $\mathcal{U}_{A^c}$ following the proof of Theorem \ref{theorem1}. Suppose the complementary subregions $A$ and $B$ are connected. Since the boundary of a 0-chain is trivial, we do not need $h$. $[k_A]$ and $[k_{A^c}]$ indicate the parity of fermions in two complementary subregions. Since we are restricted to the symmetric subspace, they are either both even or both odd. In the both-even case, $t_A([k_A])=t_B([k_{A^c}])=0$. In the both-odd case, $t_A([k_A])=t_B([k_{A^c}])=v_0$, where $v_0$ is a vertex on the boundary determined once and for all. $f_1(k_A)=l_A(k_A-t_A([k_A]))$ is a preimage of $k_A-t_A([k_A])$ under the boundary map. Such a preimage always exists because $A$ is connected. $f_2(k_{A^c})=l_B(k_{A^c}+t_B([k_{A^c}]))$ is obtained similarly. Then we can use $f_1$ and $f_2$ to construct $\mathcal{U}_A$ and $\mathcal{U}_{A^c}$ following the proof of Theorem \ref{theorem1}. In fact, this is exactly how the entanglement sum rule was proved for this Hamiltonian \cite{Swingle2013}.

\subsection{Gauging the $(3+1)$D Transverse-Field Toric Code}\label{subsec new example}
It is known that by gauging a $p$-form symmetry corresponding to an abelian group $G$, we obtain a model with $(d-p-2)$-form symmetry corresponding to $\widehat{G}$ \cite{Haegeman_2015, Luo_2024, Sakura2023}. Then we can stack two models together and add minimal coupling with $\mathcal{U}$. This allows us to “mass-produce” models satisfying the entanglement sum rule.

Take $(3+1)$D transverse-field toric code as an example. The space $X$ is a 3-dimensional $\Delta$-complex. The degrees of freedom are spin-1/2 on links. The Hamiltonian is
\eqalign{
    H = -\Delta_F \sum_{f\in\Delta^2(X)} \prod_{l\in f} \sigma^z_l - \Delta_V \sum_{v\in\Delta^0(X)} \prod_{l\ni v}\sigma_l^x-\Delta_L\sum_{l\in \Delta^1(X)}\sigma_l^x.
}
This Hamiltonian has a 1-form symmetry corresponding to products of $\sigma^x_l$ on links that cross a closed surface in the dual complex,
\eqalign{\label{1-form1}
    \{\prod_{l\in \Delta^1(X)}(\sigma^x_l)^{\phi(l)}\mid \phi\in Z^1(X;\mathbb{Z}_2)\}.
}
To gauge the 1-form symmetry, we enlarge the Hilbert space by introducing spin-1/2 degrees of freedom on 2-simplices and define the generators of gauge transformations $\{\sigma^x_l\prod_{f\ni l} \sigma^x_f\mid l\in \Delta^1(X)\}$. To gauge the Hamiltonian, we embed $H$ into the operators on the enlarged Hilbert space, project it onto the subspace that is invariant under gauge transformations. The vertex terms $\prod_{l\ni v}\sigma_l^x$ become constants and will be discarded. Then we add magnetic flux terms with $-\Delta_V$ as their factor (see \cite{Haegeman_2015, Barkeshli_2019, Choi_2025} for more details). This gives us
\eqalign{
    H_{\text{gauged}}'=-\Delta_F \sum_{f\in\Delta^2(X)} \sigma^z_f\prod_{l\in f}\sigma^z_l  -\Delta_L \sum_{l\in\Delta^1(X)}\prod_{f\ni l}\sigma^x_f- \Delta_V\sum_{c\in \Delta^3(X)} \prod_{f\in c}\sigma^z_f.
}
The gauge-invariant subspace is spanned by equal-weight superpositions of gauge-equivalent states where spins on links and plaquettes are in $\sigma^z=\pm 1$. In each set of gauge-equivalent states, there is always one state where the spin on each link is in $\sigma^z=1$. Thus, the gauge-invariant subspace is isomorphic to the tensor product of local Hilbert spaces on plaquettes. The Hamiltonian corresponding to $H_{\text{gauged}}'$ on this Hilbert space is
\eqalign{
    H_{\text{gauged}}= -\Delta_L \sum_{l\in\Delta^1(X)}\prod_{f\ni l}\sigma^x_f - \Delta_V\sum_{c\in \Delta^3(X)} \prod_{f\in c}\sigma^z_f-\Delta_F \sum_{f\in\Delta^2(X)} \sigma^z_f.
}
$H_{\text{gauged}}$ has a 1-form symmetry corresponding to products of $\sigma^z_f$ on a closed surface,
\eqalign{\label{1-form2}
    \{\prod_{f\in\Delta^2(X)}(\sigma^z_f)^{k_f}\mid k\in Z_2(X;\mathbb{Z}_2)\}.
}
\begin{remark}
    Note that $H_{\text{gauged}}$ is equivalent to the original Hamiltonian $H$ with $X$ replaced by its dual complex, $\sigma^x$ switched with $\sigma^z$, and $\Delta_F$ switched with $\Delta_L$. Thus, we have shown the $(3+1)$D transverse-field toric code is self-dual under gauging the 1-form symmetry, up to exchanging electric and magnetic variables.
\end{remark}
By stacking $H$ and $H_{\text{gauged}}$ together, we get
\eqalign{\label{decoupled2}
    H_0 =& -\Delta_F \sum_{f\in\Delta^2(X)} \prod_{l\in f} \sigma^z_l - \Delta_V \sum_{v\in\Delta^0(X)} \prod_{l\ni v}\sigma_l^x-\Delta_L\sum_{l\in \Delta^1(X)}\sigma_l^x\\
    &-\Delta_L' \sum_{l\in\Delta^1(X)}\prod_{f\ni l}\sigma^x_f - \Delta_V' \sum_{c\in \Delta^3(X)} \prod_{f\in c}\sigma^z_f-\Delta_F' \sum_{f\in\Delta^2(X)} \sigma^z_f.
}
Note that the parameters do not need to match. But we require $\Delta_V, \Delta_V'>0$ to enforce two 1-form symmetries energetically. The minimal coupling operator \eqref{unitary} is
\eqalign{
    &\mathcal{U} = \sum_{k=\partial k'\in B_1(X;\mathbb{Z}_2)}P(k)\tilde{U}(k'),\\
    &P(k)=\prod_{l\in \Delta^1(X)} \frac{1+(-1)^{k_l}\sigma^x_l}{2},\\
    &\tilde{U}(k')=\prod_{f\in \Delta^2(X)}(\sigma^z_f)^{k'_f}.
}

We now derive $\mathcal{U} H_0 \mathcal{U}^{\dagger}$ restricted to the subspace $\mathcal{H}_{\text{inv}}$ invariant under both 1-form symmetries. The 1-form symmetry \eqref{1-form2} ensures that $\mathcal{U}$ is independent of the choices of $k'$. All terms of $H_0$, except the $-\Delta_F$ and $-\Delta_L'$ terms, commute with $\mathcal{U}$. For the $-\Delta_F$ terms, note that $P(k)\prod_{l\in f}\sigma^z_l =\prod_{l\in f}\sigma^z_l P(k-\partial f)$. Thus, according to \eqref{gauged},
\eqalign{
    \mathcal{U} \prod_{l\in f} \sigma^z_l \mathcal{U}^{\dagger} = \sigma^z_f\prod_{l\in f} \sigma^z_l.
}
We omitted the projection operator $\sum_{k\in B_1(X;\mathbb{Z}_2)}P(k)$ since we are restricted to $\mathcal{H}_{\text{inv}}$. 

The $-\Delta'_L$ terms commute with $P(k)$. If $k_{l} = 0$, $P(k)$ projects onto the subspace with $\sigma^x_l=1$. Correspondingly, $\tilde{U}(k')=\prod_{f\in \Delta^2(X)}(\sigma^z_f)^{k'_f}$ acts with $\sigma^z$ on an even number of neighboring plaquettes of $l$, so $\tilde{U}(k')\prod_{f\ni l} \sigma^x_f \tilde{U}^{\dagger}(k')= \prod_{f\ni l} \sigma^x_f$. On the other hand, if $k_l=1$, $P(k)$ projects onto the subspace with $\sigma^x_l=-1$ and $\tilde{U}(k')\prod_{f\ni l} \sigma^x_f \tilde{U}^{\dagger}(k')= -\prod_{f\ni l} \sigma^x_f$. Thus, $\mathcal{U}\prod_{f\ni l} \sigma^x_f \mathcal{U}^{\dagger}=\sigma^x_l \prod_{f\ni l} \sigma^x_f $. As before, $\sum_{k\in B_1(X;\mathbb{Z}_2)}P(k)$ is omitted since we are restricted to $\mathcal{H}_{\text{inv}}$.

In summary, in the symmetric subspace $\mathcal{H}_{\text{inv}}$,
\eqalign{
    \mathcal{U} H_0 \mathcal{U}^{\dagger}=&-\Delta_F \sum_{f\in\Delta^2(X)}\sigma^z_f \prod_{l\in f} \sigma^z_l - \Delta_V \sum_{v\in\Delta^0(X)} \prod_{l\ni v}\sigma_l^x-\Delta_L\sum_{l\in \Delta^1(X)}\sigma_l^x\\
    &-\Delta_L'\sum_{l\in\Delta^1(X)}\sigma_l^x\prod_{f\ni l}\sigma^x_f - \Delta_V' \sum_{c\in \Delta^3(X)} \prod_{f\in c}\sigma^z_f-\Delta_F' \sum_{f\in\Delta^2(X)} \sigma^z_f\\
    \equiv &H.
}
$H$ is defined in $\mathcal{H}$ with the same expression as above. 

Note that the vertex operator $\prod_{l\ni v}\sigma^x_l,v\in \Delta^0(X)$ and the tetrahedron operator $\prod_{f\in c}\sigma^z_f,c\in\Delta^3(X)$ are stabilizers that commute with all other terms. We have assumed $\Delta_V,\Delta_V'>0$. Thus, the ground state of $H$ must be in the subspace where $\prod_{l\ni v}\sigma^x_l=1,\forall v\in\Delta^0(X)$, $\prod_{f\in c}\sigma^z_f=1,\forall c\in\Delta^3(X)$. If $H^1(X;\mathbb{Z}_2)$ and $H_2(X;\mathbb{Z}_2)$ are trivial, the vertex and tetrahedron operators generate all 1-form symmetry transformations, and the ground state of $H$ must be in $\mathcal{H}_{\text{inv}}$. The ground states of $H$ can be written as $\{\mathcal{U}(\ket{GS}_p\otimes\ket{GS}_{p+1})\}$, where $\{\ket{GS}_p\}$ and $\{\ket{GS}_{p+1}\}$ are symmetric ground states of the decoupled subsystems (the first and second lines of \eqref{decoupled2}). Following Theorem \ref{theorem2}, these ground states of $H$ satisfy the entanglement sum rule, provided the bipartition $(A,B)$ satisfies the Mayer--Vietoris criterion. In particular, a sufficient condition is that $A$ and $B$ are simply connected.

Note that our theorems rely only on symmetries and topology. Adding arbitrary polynomials of $\sigma^x_l$ on links and $\sigma^z_f$ on plaquettes to $H_0$ and $H$ simultaneously does not affect the above conclusions, since these terms preserve the 1-form symmetries and commute with $\mathcal{U}$.

\section{Discussion and Outlook}
In this work, we have established a generalized entanglement sum rule arising from higher-form symmetries in quantum lattice models. By constructing a minimal coupling operator $\mathcal{U}$ that (de)couples subsystems and showing its factorization under the higher-form symmetries and homology conditions, we provide a mechanism for entanglement preservation in minimally coupled systems. This explains and extends previous results on sum rules in fermion-$\mathbb{Z}_2$ gauge theories \cite{Swingle2013, Yao_2010} and enables new constructions by gauging, as demonstrated in the $(3+1)$D toric code example.

Our findings highlight the intimate connection between higher-form symmetries, homology, and entanglement structures, aligning with recent advances in understanding entanglement in gauge theories \cite{Ibieta_Jimenez_2020, Xu_2025}.

The Mayer--Vietoris sequence has played a central role in our derivation. Its connection with the entanglement entropy is not surprising, as it relates the topologies of the whole space, two subregions, and their boundary. Its interplay with symmetry and entanglement is worth further exploration. The extension of the entanglement sum rule to continuous symmetries or continuous space is also promising.

\section{Acknowledgements}
I thank Ran Luo, Zhaolong Zhang and Junkai Wang for helpful discussions.

\appendix
\section{Stokes' Theorem and Perfect Pairing}\label{appendixA}
In this section,  $n=0,1,...,d-2$. Let $I_{n+1}$ and $I_n$ be the index sets of $ \Delta^{n+1}(X)$ and $ \Delta^n(X)$. Suppose $\partial \sigma^{n+1}_i=\sum_{j\in I_n}s_{ij} \sigma^n_j$. $\forall k'\in C_{n+1}(X;\widehat{G})$ and $\phi \in C^n(X;G)$,
\eqalign{
    \braket{\partial k',\phi}=&\braket{\sum_{i\in I_{n+1}}\sum_{j\in I_n} k'_{\sigma_i^{n+1}} s_{ij}\sigma^n_j, \phi}=\sum_{i\in I_{n+1}}\sum_{j\in I_n} k'_{\sigma_i^{n+1}}(\phi(\sigma^n_j)) s_{ij},\\
    \braket{k',\delta\phi}=&\sum_{i\in I_{n+1}}k'_{\sigma_i^{n+1}}(\delta\phi(\sigma^{n+1}_i))=\sum_{i\in I_{n+1}}k'_{\sigma_i^{n+1}} (\phi(\partial\sigma_i^{n+1}))
    =\sum_{i\in I_{n+1}}\sum_{j\in I_n} k'_{\sigma^{n+1}_i}(\phi(\sigma^n_j)) s_{ij}.
}
Thus, $\braket{\partial k',\phi}=\braket{k',\delta\phi}$, which is a generalized Stokes' theorem \cite{Munkres1984}.

$\braket{\cdot,\cdot}$ induces the following homomorphism,
\eqalign{
    &\Phi: C_n(X;\widehat{G}) \rightarrow \Hom(C^n(X;G),\mathbb{R}/2\pi\mathbb{Z}),\\
    &k\mapsto \Phi(k),\ \Phi(k) (\phi)=\braket{k,\phi}.
}
If $k\neq 0$, $\exists \sigma\in\Delta^n(X),\ g\in G$  such that $k_{\sigma}(g)\neq 0\mod 2\pi$. Define $\phi^{\sigma}_g\in C^n(X;G)$ by $\phi^{\sigma}_g(\sigma')=\delta_{\sigma,\sigma'} g$, then $\Phi(k)(\phi^{\sigma}_g)=k_{\sigma}(g)\neq 0\mod 2\pi$. Thus, $\Phi$ is injective.

Given $f\in \Hom(C^n(X;G),\mathbb{R}/2\pi\mathbb{Z})$, define $k\in C_n(X;\widehat{G})$ by
\eqalign{
    k_{\sigma} (g) = f(\phi^{\sigma}_g).
}
Since $k_{\sigma}(gg')=f(\phi^{\sigma}_{gg'})=f(\phi^{\sigma}_g+\phi^{\sigma}_{g'})=f(\phi^{\sigma}_g) + f(\phi^{\sigma}_{g'})=k_{\sigma}(g) + k_{\sigma}(g')$, $k_{\sigma}\in \widehat{G}$, and $k\in C_n(X;\widehat{G})$. Then $\forall \phi \in C^n(X;G)$, $\Phi(k)(\phi) = f(\phi)$, $f=\Phi(k)$. Thus, $\Phi$ is surjective. Thus, $\Phi$ is an isomorphism. Similarly, $C^n(X;G)\cong \Hom(C_n(X;\widehat{G}),\mathbb{R}/2\pi\mathbb{Z})$. Thus, $\braket{\cdot,\cdot}$ is a perfect pairing between $C_n(X;\widehat{G})$ and $C^n(X;G)$.

For a subgroup of $C^n(X;G)$, $A^n$, define
\eqalign{
    (A^n)^{\perp}=\{k\in C_n(X;\widehat{G})\mid \forall \phi \in A^n, \braket{k,\phi}=0\mod 2\pi\}.
}
Obviously, $(A^n)^{\perp}$ is a subgroup of $C_n(X;\widehat{G})$. For $A_n\le C_n(X;\widehat{G})$,  $(A_n)^{\perp}$ is defined in a similar way. It can be verified that 
\eqalign{
    (B^n(X;G))^{\perp} = Z_n(X;\widehat{G}).
}

If $n=0$, $B^0(X;G)=\{\text{id}\}$, so $(B^0(X;G))^{\perp} = C_0(X;\widehat{G})=Z_0(X;\widehat{G})$. 

If $n>0$, $\forall k\in Z_n(X;\widehat{G}), \phi'\in C^{n-1}(X;G)$, $\braket{k, \delta \phi'}=\braket{\partial k,\phi'}=0\mod 2\pi$. So $Z_n(X;\widehat{G})\subset (B^n(X;G))^{\perp}$. $\forall k \in (B^n(X;G))^{\perp}$, $\forall \phi'\in C^{n-1}(X;G)$, $\braket{\partial k,\phi'} = \braket{k,\delta\phi'}=0\mod 2\pi$. Since $C_{n-1}(X;\widehat{G})\cong \Hom(C^{n-1}(X;G),\mathbb{R}/2\pi\mathbb{Z})$, $\partial k=0$. So $(B^n(X;G))^{\perp}\subset Z_n(X;\widehat{G})$. Thus, $(B^n(X;G))^{\perp}= Z_n(X;\widehat{G})$. 

Since $C^n(X;G)$ and $C_n(X;\widehat{G})$ are finite abelian and the pairing is perfect \cite{rotman2009homological},
\eqalign{
(Z_n(X;\widehat{G}))^\perp=B^n(X;G).}
From similar arguments, we have 
\eqalign{
&(B_n(X;\widehat{G}))^\perp=Z^n(X;G),\\
&(Z^n(X;G))^{\perp} = B_n(X;\widehat{G}).}

\section{Verification of Matter-Gauge Duality}\label{appendixB}
\begin{theorem}
    $\mathcal{U}=\bar{\mathcal{U}}$ in $\mathcal{H}_{\text{inv}}$.
\end{theorem}
\begin{proof}
\eqalign{\label{Ubar}
    \bar{\mathcal{U}}=&\sum_{\phi'=\delta\phi\in B^{p+1}(X;G)}\tilde{P}(\phi')U(\phi)\\
    =&\sum_{\phi'=\delta\phi\in B^{p+1}(X;G)}\abs{G}^{-\abs{ \Delta^{p+1}(X)}}\sum_{k'\in C_{p+1}(X;\widehat{G})} e^{-i\braket{k',\delta\phi}} \tilde{U}(k')U(\phi)\\
    =&\abs{G}^{-\abs{ \Delta^{p+1}(X)}}\sum_{\phi'=\delta\phi\in B^{p+1}(X;G)}\sum_{k'\in C_{p+1}(X;\widehat{G})} e^{-i\braket{\partial k',\phi}} \tilde{U}(k')U(\phi)\\
    =&\frac{\abs{Z_{p+1}(X;\widehat{G})}}{\abs{Z^p(X;G)}\cdot \abs{G}^{\abs{ \Delta^{p+1}(X)}}}\sum_{\phi\in C^p(X;G)}\sum_{k=\partial k'\in B_p(X;\widehat{G})} e^{-i\braket{k,\phi}} \tilde{U}(k')U(\phi)\\
    =&\frac{\abs{Z_{p+1}(X;\widehat{G})}\cdot \abs{G}^{\abs{ \Delta^p(X)}}}{\abs{Z^p(X;G)}\cdot \abs{G}^{\abs{ \Delta^{p+1}(X)}}} \sum_{k=\partial k'\in B_p(X;\widehat{G})} P(k) \tilde{U}(k')\\
    =&\frac{\abs{Z_{p+1}(X;\widehat{G})}\cdot \abs{G}^{\abs{ \Delta^p(X)}}}{\abs{Z^p(X;G)}\cdot \abs{G}^{\abs{ \Delta^{p+1}(X)}}}\mathcal{U}.
}
So what is left to prove is $\abs{Z_{p+1}(X;\widehat{G})}/\abs{Z^p(X;G)}=\abs{G}^{\abs{ \Delta^{p+1}(X)}-\abs{ \Delta^p(X)}}$. 

Define $M=C_p(X;\mathbb{Z})/B_p(X;\mathbb{Z})$. Suppose $M\cong \mathbb{Z}^{\beta_p}\oplus \bigoplus_{i=1}^s\mathbb{Z}/d_i\mathbb{Z}$ without loss of generality. The cochain groups are $C^p(X;G)=\Hom(C_p(X;\mathbb{Z}),G)$. $Z^p(X;G)$ consists of the cochains that map $B_p(X;\mathbb{Z})$ to identity. Thus, a cochain in $Z^p(X;G)$ induces a homomorphism from $M$ to $G$. Conversely, given a homomorphism from $M$ to $G$, we can construct an element of $\Hom(C_p(X;\mathbb{Z}),G)$ by first taking quotient of $B_p(X;\mathbb{Z})$, then composing with the homomorphism. Since $B_p(X;\mathbb{Z})$ is mapped to identity, it is an element of $Z^p(X;G)$. These two maps establish the isomorphism, $Z^p(X;G)\cong \Hom(M,G)$.

A homomorphism from $M\cong \mathbb{Z}^{\beta_p}\oplus \bigoplus_{i=1}^s\mathbb{Z}/d_i\mathbb{Z}$ to $G$ is uniquely determined by the image of each generator. But for the $i$th torsion part, the image of a generator can only be in $G[d_i]$ (subgroup of $G$ that satisfies $g^{d_i}=\text{id}$). Thus,
\eqalign{\label{cardinality1}
    \abs{Z^p(X;G)}=\abs{\Hom(M,G)}=\abs{G}^{\beta_p}\prod_{i=1}^s\abs{G[d_i]}.
}
Since $M=C_p(X;\mathbb{Z})/B_p(X;\mathbb{Z})$, we have an exact sequence
\eqalign{
    C_{p+1}(X;\mathbb{Z})\overset{\partial_{\mathbb{Z}}}{\longrightarrow}C_p(X;\mathbb{Z})\overset{q}{\longrightarrow} M\longrightarrow 0.
}
Since tensoring with $\widehat{G}$ is right-exact \cite{atiyah_macdonald_1969}, we have another exact sequence
\eqalign{
    C_{p+1}(X;\mathbb{Z})\otimes_{\mathbb{Z}}\widehat{G}\overset{\partial_{\mathbb{Z}}\otimes_{\mathbb{Z}}\text{id}}{\longrightarrow}C_p(X;\mathbb{Z})\otimes_{\mathbb{Z}}\widehat{G}\overset{q\otimes_{\mathbb{Z}}\text{id}}{\longrightarrow} M\otimes_{\mathbb{Z}}\widehat{G}\longrightarrow 0.
}
From the following commutative diagram
\[
\begin{tikzcd}
C_{p+1}(X;\mathbb Z)\otimes_{\mathbb{Z}} \widehat{G} \arrow[r, "\partial_{\mathbb{Z}}\otimes_{\mathbb{Z}}\mathrm{id}"] \arrow[d, "\cong"]
& C_p(X;\mathbb Z)\otimes_{\mathbb{Z}} \widehat{G} \arrow[d, "\cong"] \\
C_{p+1}(X;\widehat{G}) \arrow[r, "\partial_{\widehat{G}}"]
& C_p(X;\widehat{G})
\end{tikzcd},
\]
we have $\abs{B_p(X;\widehat{G}) }=\abs{C_p(X;\widehat{G})}/\abs{M\otimes_{\mathbb{Z}}\widehat{G}}=\abs{G}^{\abs{ \Delta^p(X)}}/\abs{M\otimes_{\mathbb{Z}} \widehat{G}}$. Therefore,
\eqalign{\label{cardinality2}
    \abs{Z_{p+1}(X;\widehat{G})}=\frac{\abs{C_{p+1}(X;\widehat{G})}}{\abs{B_p(X;\widehat{G})}}=\abs{G}^{\abs{ \Delta^{p+1}(X)}-\abs{ \Delta^p(X)}}\abs{M\otimes_{\mathbb{Z}} \widehat{G}}.
}
Since $M\otimes \widehat{G}\cong \widehat{G}^{\beta_p}\oplus \bigoplus_{i=1}^s(\widehat{G}/d_i\widehat{G})$ and $G\cong \widehat{G}$, $\abs{M\otimes_{\mathbb{Z}} \widehat{G}}=\abs{G}^{\beta_p}\prod_{i=1}^s\abs{G}/\abs{d_iG}$. Since $G[d_i]$ are exactly the elements of $G$ whose $d_i$th power is identity, we have the short exact sequence,
\eqalign{
    0\longrightarrow G[d_i]\longrightarrow G\longrightarrow d_iG\longrightarrow 0.
}
Then $\abs{G}/\abs{G[d_i]}=\abs{d_iG}$. Combining with \eqref{cardinality1} and \eqref{cardinality2}, we have
\eqalign{
    \frac{\abs{Z_{p+1}(X;\widehat{G})}}{\abs{Z^p(X;G)}}=\frac{\abs{G}^{\abs{ \Delta^{p+1}(X)}-\abs{ \Delta^p(X)}}\abs{G}^{\beta_p}\prod_{i=1}^s\abs{G}/\abs{d_iG}}{\abs{G}^{\beta_p}\prod_{i=1}^s\abs{G[d_i]}}=\abs{G}^{\abs{ \Delta^{p+1}(X)}-\abs{ \Delta^p(X)}},
}
which is what is left to prove for $\mathcal{U}=\bar{\mathcal{U}}$ in $\mathcal{H}_{\text{inv}}$.
\end{proof}
\bibliographystyle{unsrt}
\bibliography{Biblio/ref}
\end{document}